\documentclass[11pt,a4paper]{article}

\usepackage[colorlinks=true,urlcolor=Blue,citecolor=Green,linkcolor=BrickRed]{hyperref}
\usepackage[usenames,dvipsnames]{xcolor}
\usepackage[utf8]{inputenc}
\usepackage{algorithmicx}
\usepackage[noend]{algpseudocode}
\usepackage{cite}
\usepackage[OT4]{fontenc}
\usepackage{amsthm}
\usepackage{amssymb}
\usepackage{amsmath}
\usepackage{amsfonts}
\usepackage{todonotes}
\usepackage{authblk}
\usepackage{fullpage}
\usepackage{graphicx,amsmath,amsfonts}
\usepackage{algorithm}
\usepackage{multicol}
\usepackage{mathtools}
\usepackage[noend]{algpseudocode} 
\usepackage{tensor}
\usepackage{thm-restate}
\usepackage{enumerate}

 \usepackage{pifont}
\usepackage{rotating}
\usepackage{color}
 \usepackage{bbm}
 \usepackage{wasysym}
   \usepackage{stmaryrd}
   
\usepackage{algorithm}
\usepackage{algorithmicx}
\usepackage[noend]{algpseudocode} 

\usepackage{amsthm}
\usepackage{interval}
\usepackage{datetime}
\usepackage{enumerate}

\usepackage{thmtools}
\usepackage{thm-restate}
\usepackage{hyperref}
\usepackage{cleveref}

\usepackage{comment}

\title{The First Order Truth behind Undecidability of Regular Path Queries Determinacy
\footnote{Supported by the Polish National Science Centre (NCN) grant 2016/23/B/ST6/01438}}
\author{Grzegorz Głuch, Jerzy Marcinkowski, Piotr Ostropolski-Nalewaja\\
Institute of Computer Science, University of Wrocław} %\\ ul. Joliot-Curie 15, 50-383 Wroclaw, Poland

\theoremstyle{definition}
\newtheorem{definition}{Definition}[section]

\newtheorem{theorem}{Theorem}[section]
\newtheorem{lemma}[theorem]{Lemma}
\newtheorem{exercise}[theorem]{exercise}

\newtheorem{fact}[theorem]{Fact}

\newtheorem{observation}[theorem]{Observation}

\newtheorem{notation}[theorem]{Notation}
\newcommand{\ddd}[0]{\ldots}
\newcommand{\rtgd}[2]{#1 \rightarrow #2}
\newcommand{\rcrg}[1]{\rc{#1}{R}{G}}
\newcommand{\rcgr}[1]{\rc{#1}{G}{R}}
\newcommand{\rc}[3]{\rtgd{#2(#1)}{#3(#1)}}
\newcommand{\step}[0]{step}
\newcommand{\Path}[2]{#1(#2)}

\newcommand{\requests}[0]{rq}
\newcommand{\database}[0]{\mathbb{D}}
\newcommand{\history}[0]{\database_{\omega^2}}%{\mathbb{H}}
\newcommand{\histories}[0]{\Omega}
\newcommand{\pair}[1]{\langle #1 \rangle}
\newcommand{\set}[1]{\{#1\}}
\newcommand{\naturals}{\mathbb{N}}

\newcommand{\CC}[5]{(\tensor*[^{#5}_{#4}]{\mathbf{#1}}{^{#3}_{#2}})}
\newcommand{\PP}[2]{\tensor*[^{#2}_{}]{\mathbb{P}}{^{}_{#1}}}
\newcommand{\GG}[2]{\tensor*[^{#2}_{}]{\mathbb{G}}{^{}_{#1}}}
\newcommand{\LL}[3]{\tensor*[^{#3}_{}]{\mathbb{L}}{^{#2}_{#1}}}

\newcommand{\ver}[0]{V}
\newcommand{\hor}[0]{H}
\newcommand{\reduction}[0]{\mathfrak{R}}
\newcommand{\concat}{\ensuremath{+\!\!\!\!+\,}}

\begin{document}

\maketitle

\begin{abstract}
In our paper [Głuch, Marcinkowski, Ostropolski-Nalewaja, LICS ACM, 2018] we have solved an old problem stated in [Calvanese, De Giacomo, Lenzerini, Vardi, SPDS ACM, 2000] showing that query determinacy is undecidable for Regular Path Queries. Here a strong generalisation of this result is shown, and
-- we think -- a very unexpected one. We prove that no regularity is needed: determinacy remains undecidable even for finite unions of conjunctive path queries. 
\end{abstract}
 
 ------------------

\section{Introduction}
\paragraph*{Query determinacy problem (QDP)}
Imagine there is a database $\database$
we have no direct access to, and there are views of this $\database$ available to us, defined by some set of
 queries $\mathcal{Q}=\{Q_1, Q_2,\ldots Q_k\}$ (where the language of queries from $\mathcal{Q}$ is a parameter of the problem).
 And we are given another  query $Q_0$. Will we be able, regardless of $\mathbb D$,
 to compute $Q_0(\mathbb D)$
only using the views   $Q_1(\mathbb D),\ldots Q_k(\mathbb D) $? The answer depends on whether the queries in $\mathcal {Q}$ 
{\em determine\footnote{Or, using the language of \cite{CGLV00}, \cite{CGLV00a} \cite{CGLV02} and \cite{CGLV02a}, whether $\mathcal{Q}$ are \textit{lossless} with respect to $Q_0$.  }} query $Q_0$. 
Stating it more precisely, the  {\bf Query Determinacy Problem} is\footnote{More precisely, the problem
  comes in two different flavors, ``finite'' and ``unrestricted'', depending on whether the ($\clubsuit$) ``each'' ranges over finite structures only, or all structures, including infinite.}:
\vspace{1.5mm}

\noindent\fbox{%
    \parbox{\linewidth}{%
  The instance of the problem is a set of  queries $\mathcal{Q}=\{Q_1,\ldots Q_k\}$, and 
another  query $Q_0$.

 The question is whether $\mathcal{Q}$ determines $Q_0$, which means that for ($\clubsuit$) each
 two structures (database instances) ${\mathbb D}_1$ and ${\mathbb D}_2$ such that  
$Q({\mathbb D}_1)= Q({\mathbb D}_2)$ for each $Q\in \mathcal{Q}$, it also holds that  $Q_0({\mathbb D}_1) = Q_0({\mathbb D}_2)$.}%
}
\vspace{1.5mm}

\noindent
QDP is seen as a very natural static analysis 
problem in the area of database theory.  It is important for privacy (when we don't want the adversary to be able to compute the query) 
and for (query evaluation plans) optimisation (we don't need to access again the database as the given views already provide enough information).
And, as a very natural static analysis 
problem, it has  a 30 years long history as a research subject -- the oldest paper we were able to trace, where QDP is studied, is \cite{LY85}, where decidability of QDP is shown
for the case where $Q_0$ is a conjunctive query (CQ) and also the set $\mathcal Q$ consists of a single CQ.

But this is not a survey paper, so let us just point a reader interested in the history of QDP to 
 Nadime Francis's thesis \cite{F15}, which is a very good read indeed.

%%%\vspace{-2mm}  
 \subsection{The context}
%%%\vspace{-2mm}  
 
 As we said, this is a technical paper not a survey paper. But still, we need to introduce the reader to the technical context of our results.
 And, from the point of view of this introduction, there are two lines of research which are interesting: decidability problems of 
 QDP for positive fragments of  SQL (conjunctive queries and their unions) and
 for fragments of the language of  Regular Path Queries (RPQs) -- the core of most navigational graph query languages.\\

\paragraph*{QDP for fragments of SQL.}
A lot of progress was done in this area in two past decades. 
The paper  \cite{NSV06} was the first to 
present a negative result. QDP was shown there 
to be undecidable if unions of conjunctive queries are allowed in  $\mathcal{Q}$ and $Q_0$. The proof is moderately hard, but the queries  are high arity (by arity of a query we mean the number of free variables) 
and hardly can be seen as living anywhere close to database practice.

 In \cite{NSV10} it was proved that determinacy is also undecidable if the elements of $\mathcal{Q}$ are conjunctive queries  and $Q_0$ is a first order sentence 
(or the other way round). Another somehow related (although no longer contained in the first order/SQL paradigm) negative result is presented in \cite{FGZ12}: determinacy is shown there to be undecidable  if  $\mathcal{Q}$  is a DATALOG program and 
$Q_0$ is a conjunctive query. Finally,  closing the classification for the traditional relational model, it was shown in \cite{GM15} and \cite{GM16} that QDP
is undecidable for $Q_0$ and the queries in $\mathcal{Q}$ being conjunctive queries. The queries in  \cite{GM15} and \cite{GM16} are quite complicated (the Turing machine there is 
encoded in the arities of the queries), and again hardly resemble anything practical. 

On the positive side, \cite{NSV10} shows that the problem is decidable for conjunctive queries 
if each  query from $\mathcal{Q}$ has only one free variable. 

Then, in 
\cite{A11} decidability was shown for $\mathcal{Q}$ and $Q_0$ being respectively a set of conjunctive path queries and a path query.
(see Section \ref{lancuszki} for the definition).  This is an important result from the point of view of the current paper, and the proof in \cite{A11}, 
while not too difficult, is very nice -- it gives
the impression of deep insight into the real reasons why a set of conjunctive path queries determines another conjunctive path query.

The result from \cite{A11} begs for generalisations, and indeed it was generalised in \cite{P11}  to the scenario where $\mathcal{Q}$ is a set of conjunctive path queries but $Q_0$ is any conjunctive
query.\\

\paragraph*{QDP for Regular Path Queries.}
A natural extension of QDP to the graph database scenario is considered here. In this scenario, the underlying data is  modelled as graphs, in which nodes are objects, and edge labels define relationships between those objects. Querying  such graph-structured data has received much attention recently, due to numerous applications, especially for social networks.

 There are many more or less expressive query languages for  such databases (see \cite{B13}). The core of all of them (the SQL of graph databases) is
  RPQ -- the language of Regular Path Queries. RPQ queries
ask for all pairs of objects in the database that are connected by a specified
path, where the natural choice of the path specification language, as \cite{V16} elegantly explains, is the language of regular expressions.
This idea is at least 30 years old (see for example \cite{CMW87}, \cite{CM90}) and considerable effort was put to
create tools for  reasoning about regular path queries, analogous to the ones we have in the 
 traditional relational databases  context.
For example \cite{AV97} and  \cite{BFW98} investigate  decidability of the
implication problem for path constraints, which are integrity constraints used for  RPQ optimisation. Containment of
conjunctions of regular path queries has been proved decidable in
\cite{CGL98} and  \cite{FLS98}, and then, in more general setting, in \cite{JV09} and \cite{RRV15}.

Naturally, query determinacy problem has also been stated, and studied,
for Regular Path Queries. This line of
research was initiated in 
\cite{CGLV00}, \cite{CGLV00a}, \cite{CGLV02} and \cite{CGLV02a}, 
and it was in \cite{CGLV02} where the central problem of this area -- decidability of QDP for RPQ -- was first stated (called there ``losslessness for exact semantics'').

On the positive side, the previously mentioned result of Afrati \cite{A11} can be seen as a special case, where each of the regular 
languages defining the queries only consists of one word (conjunctive path queries considered in \cite{A11} constitute in fact the intersection of CQ and RPQ). 
Another positive result is presented in \cite{F17}, where ``approximate determinacy'' is shown to be decidable 
if the query $Q_0$ is (defined by) a single-word regular language (a conjunctive path query), and the languages defining the queries in $Q_0$ and  $\mathcal Q$ are  
over a single-letter alphabet.
See how difficult the analysis is here -- despite a lot of effort (the proof of the result in \cite{F17} invokes ideas from \cite{A11} but is incomparably harder) even a subcase (for a  single-word regular language) of a subcase (unary alphabet) was only understood ``approximately''.

On the negative side, in \cite{GMO18},   we showed (solving the problem from \cite{CGLV02}),  that QDP is undecidable for full RPQ.

%%%\vspace{-2mm}  
\subsection{ Our contribution} 
%%%\vspace{-2mm}  
The main result of this paper, and -- we think -- quite an unexpected one, is the following strong generalisation of the main result from \cite{GMO18}:

%%%\vspace{-2mm}
\begin{theorem}\label{i1}
Answering Determinacy question for Finite Regular Path Queries is undecidable in unrestricted and finite case (for necessary definitions see Section~\ref{preliminaries}).
\end{theorem}

%%\vspace{-2mm}
To be more precise, we show that the problem, both in the ``finite'' and 
the ``unrestricted'' versions, is undecidable.

It is, we believe, interesting to see that this negative result falls into both lines of research outlined above. Finite Regular Path Queries
are of course a subset of RPQ, where star is not allowed in the regular expressions (only concatenation and plus are), 
but on the other hand they are also Unions of Conjunctive Path Queries, so unlike general RPQs are first order queries and they also fall into the SQL category. 

Our result shows that the room for generalising the positive result from \cite{A11} is quite limited. 
What we however find most surprising is the discovery that 
it was possible to give a negative answer to the question from \cite{CGLV02}, which had been open for 15 years, without talking about RPQs at all -- undecidability 
is already in the intersection of RPQs and (positive) SQL.

\vspace{2mm}
\noindent
{\bf Remark.} \cite{B13} makes a distinction between ``simple paths semantics'' for Recursive Path Queries and ``all paths semantics''. As all the graphs we produce in this paper are acyclic (DAGs), 
all our results hold for both semantics.

\paragraph*{Organization of the paper.}
In short Section \ref{preliminaries} we introduce (very few) notions and some notations we need to use.  Sections \ref{sygnatura}--\ref{ulga} of this paper are devoted to the proof of Theorem \ref{i1}.  

In Section \ref{sygnatura} we first follow the ideas from \cite{GMO18} defining the red-green signature. Then we 
define the game of Escape and state a crucial lemma (Lemma \ref{bridge}), asserting that this game really fully characterises determinacy for Regular Path Queries. In Section \ref{uniwersalnosc} we prove this Lemma. This part follows in the footsteps of \cite{GMO18}, but with some changes: in \cite{GMO18} Escape is a solitary game, and here we prefer to see it as a two-player one.

At this point we will have the tools ready for proving Theorem \ref{i1}. In Section \ref{source} we explain what is the undecidable problem we use for our reduction, and in Section~\ref{redukcja} we present the reduction.
In Sections \ref{przewodnik} -- \ref{ulga} we use the characterisation provided by Lemma \ref{bridge} to prove correctness of this reduction. 
%Due to the space constraint, we try to  explain the main ideas and formulate all the crucial lemmas in the main body of the paper, but the proofs of the lemmas can be found in the full paper.

%%%%%%%%%%%% LONG COMMENT ON GMO18 %%%%%%%%%%%%%%%%%%%%%%%%%%%%%%%%%%%

%%%\vspace{-2mm}  
\section{How this paper relates to \cite{GMO18}}
%%%\vspace{-2mm}  

This paper builds on top of the technique developed in \cite{GMO18} to prove undecidability of QDP-RPQ for 
any languages, including infinite. 

From the point of view of the high-level architecture the two papers do not differ much. In both cases,
in order to prove that if some computational device rejects its input then the respective instance of
QDP-RPQ (or QDP-FRPQ) is positive (there is determinacy) we use a game argument. In \cite{GMO18} this game is 
solitary. The player, called {\em Fugitive}, constructs a structure/graph database (a DAG, with source $a$ and sink $b$).
He begins the game by choosing a path $\database_0$ from $a$ to $b$, which represents a word from some regular language $G(Q_0)$.
Then, in each step he must ``satisfy requests''-- if there is a path from some $v$ to $w$ in the current structure, representing a
word from some (*) regular language $Q$ then he must add a path representing a word from another 
language $Q'$ connecting  these $v$ and $w$. He loses when, in this process, a path from $a$ to $b$ from yet another language $R(Q_0)$
is created. 
 In this paper this game is replaced by a two-player game. But this is a minor difference. 
There are however two reasons why the possibility of using infinite languages is crucial in \cite{GMO18}. Due to these reasons, while, as we said, the general architecture of the proof of the negative result in this paper is the same as in \cite{GMO18},
the implementation of this architecture is almost completely different here. 

The first reason is as follows.
Because of the symmetric nature of the constraints, the language $Q$ (in (*) above) 
is always almost the same as language $Q'$ (they only have
different ``colors'', but otherwise are equal). For this reason it is not at all clear how to force {\em Fugitive}
to build longer and longer paths. This is a problem for us, as to be able to encode something undecidable we need  to 
produce structures of unbounded size. One can think that paths of unbounded length translate to potentially unbounded 
length of Turing machine tape. 

In order to solve this problem we use -- in \cite{GMO18} -- a language $G(Q_0)$. It is an infinite language and -- in his initial move --
%{\em Fugitive } could choose/commit to a path of any length he wished but no longer paths could occur later in the game. 
{\em Fugitive} could choose/commit to a path of any length he wished so that the length of the path did not need to increase in the game.
But now we only have finite languages, so also $G(Q_0)$ must be finite and we needed to invent something completely different.

The second reason is in $R(Q_0)$. This -- one can think -- is the language of ``forbidden patterns'' -- paths from 
$a$ to $b$ that {\em Fugitive} must not construct. If he does, it means that he ``cheats''. But now again, $R(Q_0)$ is finite.
So how can we use it to detect {\em Fugitive's} cheating on paths no longer than the longest one in $R(Q_0)$? This at first seemed to 
us to be an impossible task. 

But it wasn't impossible. The solution to both aforementioned problems is in the complicated machinery of languages producing edges labelled with $x$ and $y$ %--
%languages $\mathcal{Q}_{good}^{3} - \mathcal{Q}_{good}^{5}$ and $\mathcal{Q}_{good}^{12} - \mathcal{Q}_{good}^{15}$  and the constraint-checking  languages in $\mathcal{Q}_{ugly}$.

%%%%%%%%%%%%%%%%%%%%%%%%%%%%%%%%%% PRELIMINARIES

%%%\vspace{-2mm}  
\section{Preliminaries and notations}\label{preliminaries}\label{lancuszki}
%%%\vspace{-2mm}  

%%NEW
\paragraph* {Determination.} For a set of queries $\mathcal{Q}=\set{Q_1, Q_2, \ddd, Q_k}$ and another query $Q_0$, we say that $\mathcal{Q}$ {\em determines} $Q_0$
if and only if: $ \forall_{\database_1,\database_2}\;  \mathcal{Q}(\database_1) = \mathcal{Q}(\database_2) \rightarrow Q_0(\database_1) = Q_0(\database_2),$
where $\mathcal{Q}(\database_1) = \mathcal{Q}(\database_2)$ is defined as $\forall_{Q \in \mathcal{Q}}\; Q(\database_1) = Q(\database_2)$. {\em Query determinacy} comes in two versions: {\em unrestricted} and {\em finite}, depending on whether we allow or disallow infinite structures $\database$ to be considered. When we speak about {\em determinacy } without specifying explicitly it's version,  we assume the {\em unrestricted} case of the problem.
%%WEN

\paragraph* {Structures.}
When we say ``structure'' we always mean a directed graph with edges labelled with letters from some signature/alphabet $\Sigma$. In other words every structure we consider is a relational structure $\database$ over some signature $\Sigma$ consisting of binary predicate names. Letters $\mathbb{D}$, $\mathbb{M}$, $\mathbb{G}$ and $\mathbb{H}$ are used to denote structures. $\Omega$ is used for a set of structures.
Every structure we consider will contain two  distinguished constants $a$ and $b$.
 For two structures ${\mathbb G}$ and ${\mathbb G'}$ over $\Sigma$, with sets of vertices $V$ and $V'$, a function $h:V \rightarrow V'$ is called a 
homomorphism if for each two vertices $\pair{u,v}$ connected by an edge with label $e \in \Sigma$ in $\mathbb{G}$ there is an edge connecting $\pair{h(u),h(v)}$, with the same label $e$, in $\mathbb{G'}$.

\paragraph* {Conjunctive path queries.}
Given a set of binary predicate names $\Sigma$ and a word $w = a_1a_2\ddd a_n$ over $\Sigma^*$ we 
define a (conjunctive) path query $ w(v_0, v_n)$ as a conjunctive query:

\begin{center}
$\exists_{v_1,\ddd,v_{n-1}} a_1(v_0, v_1) \wedge  a_2(v_1, v_2) \wedge \ddd a_n(v_{n-1}, v_n).$
\end{center}

We use the notation $w[v_0,v_n]$ to denote the canonical structure (``frozen body'') of query $w(v_0, v_n)$ --
the structure consisting of elements $v_0,v_1,\ldots v_n$ and atoms $a_1(v_0, v_1),$ $a_2(v_1, v_2),$ $\ldots$ $a_n(v_{n-1}, v_n)$.

\paragraph*{Regular path queries.}
For a regular language $Q$ over $\Sigma$ we define a query, which is also denoted by $Q$, as $Q(u, v) = \exists_{w \in Q} w(u, v)$

In other words such a query $Q$ looks for a path in the given graph labelled with any word from $Q$ and returns the endpoints of that path. Clearly, if $Q$ is a finite regular language (finite regular path query), then $Q(u,v)$ is a union of conjunctive queries.

We use letters $Q$ and $L$ to denote regular languages and $\mathcal{Q}$ and $\mathcal{L}$ to denote sets of regular languages. 
The notation $Q(\database)$ has the natural meaning: $Q(\database) = \set{\pair{u,v}\, |\, \database \models Q(u,v)}$.

%%%%%%%%%%%%%%%%%%%%%%%%%%%%%%%%%%%%%%%%%%%%%%%%%%%%%%%%%%%%%%%%%%%%%%%%%%%%%%%%%%%%%%%%%%%%%%%%%%%%%
%%%\vspace{-2mm}  
\section{Red-Green Structures and Escape}\label{sygnatura}
%%%\vspace{-2mm}  
%%%%%%%%%%%%%%%%%%%%%%%%%%%%%%%%%%%%%%%%%%%%%%%%%%%%%%%%%%%%%%%%%%%%%%%%%%%%%%%%%%%%%%%%%%%%%%%%%%%%%
In this section we will provide crucial tool for our proof. First we will introduce red-green structures. Such a structure will consist of two structures each with distinct colour. One can think that we take two databases and then colour one green and another red and then look at them as a whole. This notion is very useful for two coloured {\em Chase} technique from \cite{GM15} and \cite{GM16} that has evolved into  Game of Escape in \cite{GMO18} and is also present in this paper.

\subsection{Red-green signature and Regular Constraints}
%%%\vspace{-2mm}  
\label{red-green-section}
For a given alphabet (signature) $\Sigma$ let $\Sigma_G$ and $\Sigma_R$ be two copies of $\Sigma$ one written with ``green ink'' and another with ``red ink''. Let $\bar\Sigma = \Sigma_G \cup \Sigma_R$.

For any word $w$ from $\Sigma^*$ let $G(w)$ and $R(w)$ be copies of this word written in green and red respectively. For a regular language $L$ over $\Sigma$ let $G(L)$ and $R(L)$ be copies of this same regular language but over $\Sigma_G$ and $\Sigma_R$ respectively.
Also for any structure $\database$ over $\Sigma$ let $G(\database)$ and $R(\database)$ be copies of this same structure $\database$ but with labels of edges recolored to green and red respectively.
 For a pair of regular languages $L$ over $\Sigma$ and $L'$ over $\Sigma'$ we define the \textit{Regular Constraint (RC)} $\rtgd{L}{L'}$ as a formula
$\forall_{u,v} L(u,v) \Rightarrow L'(u,v).$

We use the notation $\database \models t$ to say that an RC $t$ is satisfied in $\database$. Also, we write $\database \models T$ for a set $T$ of RCs when for each $t\in T$ it is true that $\database \models t$.

For a graph $\database$ and an RC $t=\rtgd{L}{L'}$ let $\requests(t, \database)$ (as ``requests'') be the set of all triples $\pair{u,v,\rtgd{L}{L'}}$ such that $\database \models L(u,v)$ and $\database \not\models L'(u,v)$. For a set $T$ of RCs by 
$\requests(T, \database)$ we mean the union of all sets $\requests(t, \database)$ such that 
$t \in T$. Requests are there in order to be satisfied:

%%%\vspace{-2mm}
\begin{algorithm}[H]
\begin{algorithmic}[1]
\Statex \textbf{function} \textit{Add}
\Statex \textbf{arguments}:
\Statex \begin{itemize}
			\item Structure $\database$
            \item RC $\rtgd{L}{L'}$
			\item pair $\pair{u,v}$ such that $\pair{u,v,\rtgd{L}{L'}} \in \requests(\rtgd{L}{L'}, \database)$
		\end{itemize}
\Statex \textbf{body}:
\State Take a word $w = a_0 a_1 \ddd a_n$ from $L'$ and create a new path \mbox{$w[u,v]=a_0(u, s_1),a_1(s_1,s_2),\ddd,a_n(s_{n-1},v)$} where $s_1,s_2,\ddd,s_{n-1}$ are {\bf new} vertices
\State \Return $\database \cup w[u,v]$.
\end{algorithmic}
\label{alg:local_improvements}
\end{algorithm}

%%%\vspace{-2mm}
Notice that the result $Add(D, \rtgd{L}{L'}, \pair{u,v})$ depends on the choice of $w\in L'$. So the procedure is non-deterministic.

For a  regular language ${L}$ we define 
${L}^\rightarrow = \rcgr{L}$ and ${L}^\leftarrow = \rcrg{L}$. 
All regular constraints we are going to consider are either ${L}^\rightarrow$ or ${L}^\leftarrow$.
For a  regular language $L$ we define
$L^\leftrightarrow=\{L^\rightarrow, L^\leftarrow\}$ and for a set $\mathcal{L}$ of regular languages we define:
${\mathcal L}^\leftrightarrow = \bigcup_{L \in \mathcal{L}} L^\leftrightarrow.$

Requests of the form $\langle u,v,t\rangle$ for some RC $t$ of the form $L^\rightarrow$ ($L^\leftarrow$) are {\em generated by} $G(L)$ (resp. {\em by}  $R(L))$. Requests that are generated by $G(L)$ or $R(L)$ are said to be {\em generated by} $L$. 

The following lemma is straightforward to prove and characterises determinacy in terms of regular constraints:

%%%\vspace{-2mm}  
\begin{lemma}
\label{lm-det-struct}
A set $\mathcal{Q}$ of regular path queries over $\Sigma$ does not determine (does not finitely determine) a regular path query  $Q_0$, over the same alphabet,  if and only if there exists a structure $\mathbb M$ (resp. a finite structure) and a pair of vertices $u,v \in {\mathbb M}$ such that ${\mathbb M} \models \mathcal{Q}^\leftrightarrow$ and ${\mathbb M} \models {(G(Q_0))}(u,v)$ but   ${\mathbb M}\not\models {(R(Q_0))}(u,v-)$.
\end{lemma}
%%%\vspace{-2mm}  

Any structure ${\mathbb M}$, as above, will be called a \textit{counterexample}. One can think of $\mathbb{M}$ as a pair of structures, being green and red parts of $\mathbb{M}$. Note that those two structures both agree on $\mathcal{Q}$ but don't on $Q_0$, thus proving that $\mathcal{Q}$ does not determine $Q_0$.

%%%\vspace{-2mm}  
\subsection{The game of Escape}\label{game}
%%%\vspace{-2mm}  
%NEW
Here we present the essential tool for our proof. One can note that the game of Escape is very simmilar to the well known {\em Chase} technique. This is indeed the case, as one can think about RCs as of {\em Tuple Generating Dependencies} (TGDs) from the {\em Chase}. Divergence from standard {\em Chase} comes from ``nondeterminism'' that is inherent part of RCs (request can be satisfied by any word from a language) phenomenon not present in TGDs.

An instance Escape($Q_0$, $\mathcal{Q}$) of a game called \textit{Escape}, played by two players called  \textit{Fugitive} and  \textit{Crocodile}, is:
\begin{itemize}
\item A finite regular language $Q_0$ of {\em forbidden paths} over $\Sigma$.
\item A set $\mathcal{Q}$ of finite regular languages over $\Sigma$,
\end{itemize}

The rules of the game are:
\begin{itemize}
	\item First {\em Fugitive} picks the \textit{initial position} of the game as $\database_0 = (G(w))[a,b]$ for some $w \in Q_0$.

\item Suppose $\database_{\beta}$ is the current  position of some play before move $\beta+1$ and let $S_\beta=\requests(\mathcal{Q}^\leftrightarrow, \database_\beta)$. Then, in move $\beta+1$, {\em Crocodile} picks one request $\pair{u,v,t} \in S_\beta$ and then {\em Fugitive} can move to any position of the form:
$$\database_{\beta+1} := Add(\database_\beta, t, \pair{u,v})$$

\item For a limit ordinal $\lambda$ the position  $\database_\lambda $ is defined as 
$\bigcup\limits_{\beta < \lambda} \database_\beta$.

\item  If $\requests(\mathcal{Q}^\leftrightarrow, \database_i)$ is empty then for each $j > i$ the structures $\database_{j}$ and $\database_i$ are equal.

\item {\em Fugitive} loses when for a \textit{final position} $\history = \bigcup\limits_{\beta<\omega^2} \database_{\beta}$ it is true that $\history \models (R(Q_0))(a,b)$, otherwise he wins. Obviously if there is some $\beta < \omega^2$ such that $\database_{\beta} \models (R(Q_0))(a,b)$ then the result of the game is already known ({\em Fugitive} loses), but technically the game still proceeds.

\end{itemize}

Notice that we want the game to last $\omega^2$ steps. This is not really crucial
(if we were careful $\omega$ steps would be enough)
but costs nothing and will simplify presentation in Section~\ref{straszna}.

Obviously, different strategies of both players may lead to different final positions. 

Now we can state the crucial Lemma, that connects the game of Escape and QDP-RPQ:

%%%\vspace{-2mm}  
\begin{lemma}\label{bridge}
For an instance of QDP-RPQ consisting of regular language $Q_0$ over $\Sigma$ and a set of regular languages $\mathcal{Q}$ over $\Sigma$ the two conditions are equivalent:
\begin{enumerate}[(i)]
\item $\mathcal{Q}$ does not determine $Q_0$,
\item {\em Fugitive} has a winning strategy in Escape($Q_0$, $\mathcal{Q}$).
\end{enumerate}
\end{lemma}
%%%\vspace{-2mm}  

%Appendix, Section~\ref{uniwersalnosc}.

We should mention here that all the notions of Section~\ref{sygnatura} are similar to those of \cite{GMO18} but are not identical. The most notable difference is in the definition of the game of Escape, as it is no longer a solitary game, as it was in \cite{GMO18}.
%no longer unlike it was in \cite{GMO18}.

This makes the analysis slightly harder here, but pays off in Sections~\ref{przewodnik} -- \ref{ulga}.
%%%\vspace{-2mm}  
\subsection{Universality of Escape (Proof of Lemma~\ref{bridge}}\label{uniwersalnosc})
%%%\vspace{-2mm}  
\noindent
It is clear that $(i)\Leftarrow(ii)$ is true. All we need is to use the final position of a play won by {\em Fugitive} as the counterexample for determinacy as in Lemma~\ref{lm-det-struct}. 
But the other direction is not at all obvious. Notice that it could {\em a priori} happen that, while
some counterexample exists, it is some terribly complicated structure which {\em Fugitive}
can not force {\em Crocodile} to reach as a final position in a play of the game of Escape.

We will denote the set of all final positions reachable (by any sequence of moves of both players) from an initial position $\database_0$, for a set of regular languages $\mathcal{L}$, as $\histories(\mathcal{L}^\leftrightarrow, \database_0)$.

\begin{lemma} \label{universal}
Suppose 
structures $\database_0$ and $\mathbb M$ over $\bar\Sigma$ are such that there exists 
a homomorphism $h_0:\database_0\rightarrow\mathbb M$. Let $T$ be a set of RCs and suppose  ${\mathbb M} \models T$.
Then (regardless of {\em Crocodile's} moves) {\em Fugitive} can reach some final position $\history\in\histories(T, \database_0)$ such that there exists a homomorphism $h$ from $\history$ to $\mathbb{M}$.
\end{lemma}
\begin{proof}

Next lemma provides the induction step for the proof of  Lemma~\ref{universal}.

Let us define $\step$ as an arity four relation such that $\pair{\database,\database', T, r} \in \step$  when $\database'$ can be the result of one move of {\em Fugitive}, in position  $\database$, in the game of Escape with set of RCs $T$ and a particular request 
$r \in \requests(T, \database)$
picked by {\em Crocodile}.

\begin{lemma}
\label{lm-universal}
Let $\database_\beta$, $\mathbb M$ be structures over $\bar\Sigma$ and $h_\beta :\database_\beta\rightarrow\mathbb{M}$ be a homomorphism. Suppose that for a set $T$ of RCs it is true that ${\mathbb M} \models T$. Then for every $r\in \requests(T, \database_\beta)$ there exists some structure $\database_{\beta+1}$ such that $\step(\database_\beta,\database_{\beta+1}, T, r)$ and such that there exists a homomorphism $h_{\beta+1}:\database_{\beta+1}\rightarrow\mathbb{M}$ such that $h_\beta \subseteq h_{\beta+1}$.
\end{lemma}

\begin{proof}
Let $r = \pair{u,v,\rtgd{X}{Y}}$ for some $u,v\in \database_\beta$ and let $u' = h_\beta(u)$ and $v' = h_\beta(v)$. Note note that $X$ is either $G(L)$ and $Y$ is $R(L)$ or the converse.
Since  $\database_\beta\models X(u,v)$ and since $h_\beta$ is a homomorphism we know that 
${\mathbb M} \models X(u',v')$.
But ${\mathbb M} \models T$ so there is also ${\mathbb M} \models Y(u',v')$ and thus for some $a_1a_2\ddd a_n \in Y$ there is a path $p' = a_1(u',s'_1), \\ a_2(s'_1,s'_2)\ddd  a_n(s'_{n-1}, v')$ in ${\mathbb M}$. Let $\database_\beta'$ be a structure created by adding to $\database_\beta$ the new path $p = a_1(u,s_1),\\a_2(s_1,s_2),\ddd a_n(s_{n-1}, y)$ (with $s_i$ being new vertices). Let $h_\beta' = h_\beta \cup \set{\pair{s_i, s_{i}'} | i \in [n-1]}$. It is easy to see that $\database_\beta'$ and $h_\beta'$ are requested $\database_{\beta+1}$ and $h_{\beta+1}$.
\end{proof}

Now we consider the limit case. Let $\lambda$ be a limit ordinal such that $\lambda \leq \omega^2$. By definition we know that $\database_\lambda = \bigcup_{\beta < \lambda} \database_\beta$. Now we need to construct a homomorphism $h_\lambda$. Let $h_\lambda := \bigcup_{\beta < \lambda}h_\beta$. Observe that such $h_\lambda$ is a valid homomorphism from $\database_\lambda$ to $\mathbb{M}$. 

This along with Lemma~\ref{lm-universal} proves that $\database_{\omega^2}$ and $h_{\omega^2}$ are as required by Lemma~\ref{universal}.\end{proof}

Now we will prove the (i)$\Rightarrow$(ii) part of Lemma~\ref{bridge}.

Assume (i). Let ${\mathbb M}$ be a counterexample as in
Lemma~\ref{lm-det-struct}. Let  $a,b$ and $w\in Q_0$ be such that ${\mathbb M} \models (G(w))(a,b)$ and ${\mathbb M} \not\models (R(Q_0))(a,b)$. Applying Lemma~\ref{universal} to $\database_0 = G(w)[a,b]$ and 
to $\mathbb M$ we know that {\em Fugitive} (regardless of {\em Crocodile's} moves) can reach some winning final position $\history$ such that there is homomorphism from $\history$ to ${\mathbb M}$. It is clear that $\history \not\models (R(Q_0))(a,b)$ as we know that ${\mathbb M} \not\models (R(Q_0))(a,b)$. This shows that $\history$ is indeed a winning final position.

This concludes the proof of the Lemma \ref{bridge}.%\end{proof}
%%%%%%%%%%%%%%%%%%%%%%%%%%%%%%%%%%%%%%%%%%%%%%%%%%%%%%%%%%%%%%%%%%%%%%%%%%%%%%%%%%%%%%%%%%%%%%%%

%%%%%%%%%%%%%%%%%%%%%%%%%%%%%%%%%%%%%%%%%%%%%%%%%%%%%%% THE REDUCTION %%%%%%%%%%%%%%%%%%%%%%%%%%%%%%%%%%%%%%%%

%%%\vspace{-2mm}  
\section{Source of undecidability}\label{source}
%%%\vspace{-2mm}  
In this section we will define tiling problem that we will reduce to QDP-FRPQ. In order to prove undecidability for both finite and unrestricted case we will build our tiling problem upon notion of {\em recursively inseparable sets}. 

\begin{definition}[\textbf{Recursively inseparable sets}]
Sets $A$ and $B$ are called \textit{recursively inseparable} when each set $C$, called a \textit{separator}, such that $A \subseteq C$ and $B \cap C = \emptyset$, is undecidable \cite{R87}.
\end{definition}
It is well known that:
%%%\vspace{-2mm}  
\begin{lemma}
\label{turing-inseparable}
Let $T$ be the set of all Turing Machines. Then sets $T_{acc} =  \set{\phi \in T | \phi(\varepsilon) = 1}$ and $T_{rej} = \set{\phi \in T | \phi(\varepsilon) = 0}$ are recursively inseparable. By $\phi(\varepsilon)$ we mean the returned value of the Turing Machine $\phi$ that was run on an empty tape.
\end{lemma}
%%%\vspace{-2mm}  

\begin{definition}[\textbf{Square Grids}]
For a $k\in \naturals$ let $[k]$ be the set $\set{i \in\naturals | 0 \leq i \leq k}$.
A square grid is a directed graph $\pair{V,E}$ where $V = [k] \times [k]$ for some natural $k > 0$ or $V = \naturals \times \naturals$. $E$ is defined as $E(\pair{i,j},\pair{i+1, j})$ and $E(\pair{i,j},\pair{i,j+1})$ for each relevant $i,j \in \naturals $.
\end{definition}

\begin{definition}[\textbf{Our Grid Tiling Problem (OGTP)}]\label{ogtp}
An instance of this problem is a set of {\bf shades}  $\mathcal{S}$ having at least two elements (\textbf{gray},\textbf{black} $\in \mathcal{S}$) and a set $\mathcal{F} \subseteq \{V,H\} \times \mathcal{S} \times \{V,H\} \times \mathcal{S} $ of forbidden pairs $\pair{c,d}$ where $c,d \in \{V,H\} \times \mathcal{S}$. Let the set of all these instances be called $\mathcal{I}$.
\end{definition}

\begin{definition}
\label{shading}
A \textit{proper shading\footnote{We would prefer to use the term ``coloring'' instead, but we already have colors, red and green, and they shouldn't be confused with shades.}} is an assignment of shades to edges of some square grid $\mathbb{G}$ (see Figure~\ref{fig:fig0}) such that:
\begin{itemize}
\item[(a1)] each horizontal edge of $\mathbb{G}$ has a label from $\{H\} \times \mathcal{S}$.
\item[(a2)] each vertical edge of $\mathbb{G}$ has a label from $\{V\} \times \mathcal{S}$.
\item[(b1)] bottom-left horizontal edge is shaded \textbf{gray}\footnote{We think of $(0,0)$ as the bottom-left corner of a square grid. By `right' we mean a direction of the increase of the first coordinate and by `up' we mean a direction of increase of the second coordinate.}.
\item[(b2)] upper-right vertical edge (if it exists) is shaded \textbf{black}.
\item[(b3)] \textbf{G} contains no forbidden paths of length $2$ labelled by $\langle c,d \rangle \in \mathcal{F}$.
\end{itemize}

\begin{figure}
\centering
\includegraphics[width=0.4\linewidth]{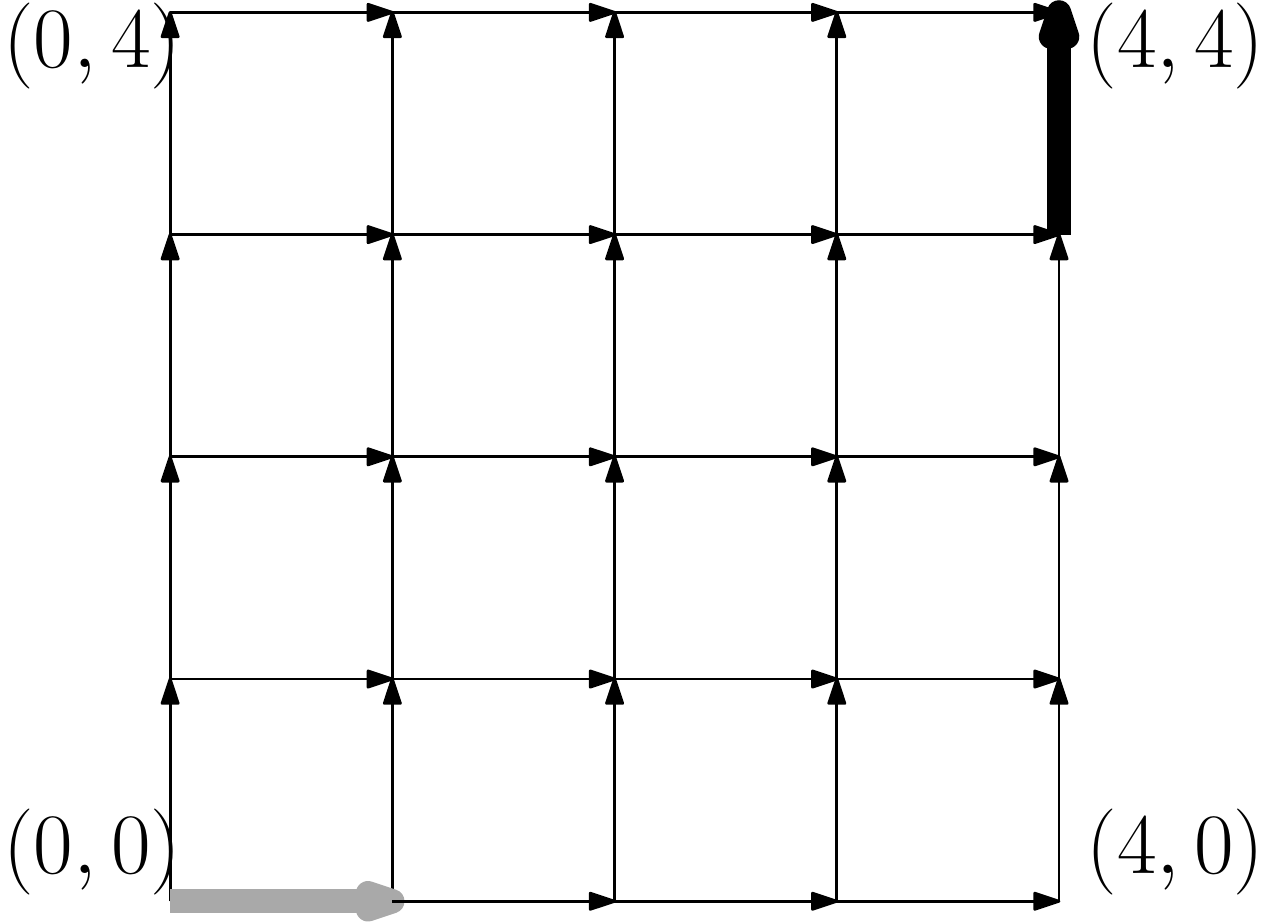}
\caption{\label{fig:fig0} Finite square grid.}
\end{figure}

We define two subsets of instances of OGTP:

\begin{itemize}
    \item[$\mathcal{A} =$] $\{ I \in \mathcal{I} | $there exists a \textit{proper shading} of some finite square grid $\}$.
    \item[$\mathcal{B} =$] $\{ I \in \mathcal{I} | $there is no \textit{proper shading} of any square grid $\}$.
\end{itemize}
\end{definition}

By a standard argument, using Lemma~\ref{turing-inseparable}, one can show that:

%%%\vspace{-2mm}  
\begin{lemma}
\label{turing-grid-lemma}
Sets $\mathcal{A}$ and $\mathcal{B}$ of instances of OGTP are recursively inseparable. 
\end{lemma}
%%%\vspace{-2mm}  

In Section~\ref{redukcja} we will construct a function $\reduction$ ($\reduction$ like $\reduction$eduction) from $\mathcal{I}$ (instances of OGTP) to instances of QPD-FRPQ that will satisfy the following:

%%%\vspace{-2mm}  
\begin{restatable}[]{lemma}{wazny}\label{main-lemma}\label{poprawnoscredukcji}
For any instance $I = \pair{\mathcal{S},\mathcal{F}}$ of OGTP and for $\pair{\mathcal{Q},Q_{0}} = \reduction(I)$:
\begin{enumerate}[(i)]
\item \label{black} If $I\in\mathcal{A}$ then $\mathcal{Q}$ does not finitely determine $Q_{0}$.
\item \label{white} If $I\in\mathcal{B}$ then $\mathcal{Q}$ determines $Q_{0}$.
\end{enumerate}
\end{restatable}

%NEW
Now the need for notion of recursive inseparability should be clear. Imagine, for the sake of contradiction, that we have an algorithm $ALG$ deciding determinacy in either finite or unrestricted case. Then, in both cases, algorithm $ALG \circ \reduction$ would separate $\mathcal{A}$ and $\mathcal{B}$, which contradicts the recursive inseparability of $\mathcal{A}$ and $\mathcal{B}$ (Lemma~\ref{turing-grid-lemma}). That will be enough to prove Theorem~\ref{i1}.

%%%\vspace{-2mm}  
\section{The function $\reduction$}\label{redukcja}
%%%\vspace{-2mm}  

Now we define a function $\reduction$, as specified in Section~\ref{source}, from the set of instances of OGTP to the set of instances of QDP-FRPQ. Suppose an instance $\langle \mathcal{S},\mathcal{F}\rangle$ of OGTP is given. We will construct an instance $\langle \mathcal{Q}, Q_0\rangle = \reduction(\pair{\mathcal{S},\mathcal{F}})$ of QDP-FRPQ.  The edge alphabet (signature) will be: $\Sigma = \{\alpha^{C}, \alpha^{W}, x^{C}, x^{W}, y^{C}, y^{W}, \$^{C}, \$^{W}, \omega\} \cup \Sigma_{0}$
where $\Sigma_{0} = \{\textbf{A},\textbf{B}\} \times \{\hor, \ver\} \times \{W,C\} \times \mathcal{S}$. We think of
$H$ and $V$ as {\bf orientations} -- \textit{Horizontal} and \textit{Vertical}. $W$ and $C$ stand for \textit{warm} and \textit{cold}.
It is worth reminding at this point that relations from $\bar\Sigma$ will -- apart from shade, orientation and temperature -- have also a {\bf color}, red or green.

\begin{notation}
\label{notation-cc}
%We  use the following notation for elements of $\Sigma_{0}$:
We will denote $(\mathbf{l},o,t,s) \in \Sigma_{0}$ as $\CC{l}{o}{t}{s}{\empty}$.

Symbol $\bullet$ and empty space are to be understood as wildcards. %When writing $\CC{p}{q}{r}{s}{t}$ with some of $\mathbf{p},q,r,s,t$ omitted we mean a subset of $\Sigma_{0}$ of elements that match the pattern $\CC{p}{q}{r}{s}{t}$. 
This means, for example,  that $\CC{A}{\hor}{}{s}{\empty}$ denotes the set $\{\CC{A}{\hor}{W}{s}{\empty}, \CC{A}{\hor}{C}{s}{\empty} \}$ and $\CC{\bullet}{\hor}{W}{s}{\empty}$ denotes $\{\CC{A}{\hor}{W}{s}{\empty}, \CC{B}{\hor}{W}{s}{\empty} \}$. 

Symbols from $\CC{\bullet}{\empty}{W}{\empty}{\empty}$ and $\{\alpha^{W}, x^{W}, y^{W}, \$^{W}\}$ will be called \textbf{warm} and symbols from $\CC{\bullet}{\empty}{C}{\empty}{\empty}$ and $\{\alpha^{C}, x^{C}, y^{C}, \$^{C}\}$ will be called \textbf{cold}.
\end{notation}

\noindent
Now we  define $\mathcal{Q}$ and $Q_{0}$. 
Let $\mathcal{Q}_{good}$ be a set of 15 languages:
%%%\vspace{-4mm}
\begin{multicols}{2}
\begin{enumerate}
\item $\omega$
\item $\alpha^{C} + \alpha^{W}$
\item $x^{C} + x^{W}$
\item $y^{C} + y^{W}$
\item $\$^{C} + \$^{W}$
\item $\CC{B}{\ver}{C}{\empty}{\empty} + \CC{B}{\ver}{W}{\empty}{\empty}$
\item $\CC{B}{\hor}{W}{\empty}{\empty} + \CC{B}{\hor}{C}{\empty}{\empty}$
\item $\CC{A}{\ver}{W}{\empty}{\empty} + \CC{A}{\ver}{C}{\empty}{\empty}$
\item $\CC{A}{\hor}{C}{\empty}{\empty} + \CC{A}{\hor}{W}{\empty}{\empty}$
\item $\CC{B}{\hor}{W}{\empty}{\empty}\CC{A}{\ver}{W}{\empty}{\empty} + \CC{B}{\ver}{C}{\empty}{\empty}\CC{A}{\hor}{C}{\empty}{\empty}$
\item $\CC{A}{\hor}{C}{\empty}{\empty}\CC{B}{\ver}{C}{\empty}{\empty} + \CC{A}{\ver}{W}{\empty}{\empty}\CC{B}{\hor}{W}{\empty}{\empty}$
\item $x^{C} \hspace{-1mm}\left(\CC{A}{\hor}{C}{\empty}{\empty}\hspace{-1mm} + \hspace{-1mm}\CC{B}{\hor}{C}{\empty}{\empty} \hspace{-1mm}+ \hspace{-1mm}\CC{A}{\ver}{C}{\empty}{\empty} \hspace{-1mm}+ \hspace{-1mm}\CC{B}{\ver}{C}{\empty}{\empty} \right) \hspace{-1mm}+\hspace{-1mm}
	   x^{C} +
       x^{W}$
\item $\left(\CC{A}{\hor}{C}{\empty}{\empty} \hspace{-1mm}+ \hspace{-1mm}\CC{B}{\hor}{C}{\empty}{\empty} \hspace{-1mm}+\hspace{-1mm} \CC{A}{\ver}{C}{\empty}{\empty} \hspace{-1mm}+\hspace{-1mm} \CC{B}{\ver}{C}{\empty}{\empty} \right)y^{C} \hspace{-1mm}+\hspace{-1mm}
	   y^{C} \hspace{-1mm}+\hspace{-1mm}
       y^{W}$
\item $x^{W} + x^{C} + x^{C}\CC{A}{\hor}{C}{\empty}{\empty}\CC{B}{\ver}{C}{\empty}{\empty}$
\item $y^{W} + \$^{C} + \CC{A}{\hor}{C}{\empty}{\empty}\CC{B}{\ver}{C}{\empty}{\empty}y^{C} + \CC{B}{\ver}{C}{\empty}{\empty}y^{C}$
\end{enumerate}
\end{multicols}

\noindent
\hspace{-1mm}
Let $\mathcal{Q}_{bad}$ be a set of languages:
\begin{enumerate}
\item $\alpha^{W} x^{W} \CC{\bullet}{d}{W}{s}{\empty} \CC{\bullet}{d'}{W}{s'}{\empty} y^{W} \omega$ for each forbidden pair $\langle(d,s), (d',s')\rangle \in \mathcal{F}$.
\item $\alpha^{W} x^{W} \CC{B}{\ver}{W}{shade}{\empty} \$^{W} \omega$ for each $shade \in \mathcal{S} \setminus \{black\}$.
\end{enumerate}

\noindent
Finally, let $\mathcal{Q}_{ugly}$ be a set of languages:  1. $\alpha^{C}\Sigma^{\leq 4}\CC{\bullet}{\empty}{W}{\empty}{\empty}\Sigma^{\leq 4}\omega$, 2. $\alpha^{W}\Sigma^{\leq 4}\CC{\bullet}{\empty}{C}{\empty}{\empty}\Sigma^{\leq 4}\omega$ and 3.~$\alpha^C x^C \CC{B}{\ver}{C}{\empty}{\empty}\CC{B}{\ver}{C}{\empty}{\empty}y^C\omega$. Where $\Sigma^{\leq 4}$ is language over $\Sigma$ of words shorter than $5$.

\noindent

We write ${Q}_{good}^{i}, {Q}_{bad}^{i}, {Q}_{ugly}^{i}$ to denote the i-th language of the corresponding group. Now we can define: $\mathcal{Q} := \mathcal{Q}_{good} \cup \mathcal{Q}_{bad} \cup \mathcal{Q}_{ugly}$

The sense of the construction will (hopefully) become clear later, but we can already say something about the structure of the defined languages.

%NEW
Languages from $Q_{good}$ serve as a building blocks during the game of Escape, they are used solely to build a shaded square grid structure, that will correspond to some tiling of square grid. On the other hand the languages from $Q_{bad}$ and $Q_{ugly}$ serve as a ``stick'' during play. It will become clear later that whenever {\em Fugitive} has to answer a request generated by one of languages in $Q_{bad}$ or $Q_{ugly}$ it is already too late for him to win. More precisely every request generated by $Q_{bad}$ is due to { \em Fugitive} assigning shades to grid in a forbidden manner and leads to {\em Fugitive's} swift demise. Similarly requests generated by $Q_{ugly}$ make {\em Fugitive} to use ``right'' building blocks from $Q_{good}$, by forcing that edge is red if and only if it is {\em warm} (this will be thoroughly explained in Section~\ref{principles-section}), or else he will loose. Third language from $Q_{ugly}$ ensures that there won't be two consecutive $\CC{B}{\ver}{C}{\empty}{\empty}$ symbols in the structure, it is a feature used exclusively in proof of Lemma~\ref{Pk}.

Finally, define  $Q_{start} := \alpha^{C} x^{C} \CC{A}{\hor}{C}{gray}{\empty}\CC{B}{\ver}{C}{\empty}{\empty}  y^{C} \omega$, and let: 

\hspace{-2mm}
$$Q_{0} := Q_{start} + \bigoplus_{L \in \mathcal{Q}_{ugly}} {L} + \bigoplus_{L \in \mathcal{Q}_{bad}} {L}$$
\hspace{-2mm}

%NEW
Intuitively, we need to put the languages $Q_{bad}$ and $Q_{ugly}$ into $Q_{0}$ so that they can serve as aforementioned ``sticks'' such that whenever {\em Fugitive} makes a mistake he loses. Recall that {\em Fugitive} loses when he writes a red word from $Q_0$ connecting the beginning and the end of the starting structure. On the other side, this creates a problem regarding the starting structure (now {\em Fugitive} can pick a word from $Q_{bad}$ or $Q_{ugly}$ to start from instead of a word from $Q_{start}$), but fortunately it can be dealt with and will be resolved in Section~\ref{principles-section}.

%%%%%%%%%%%%%%%%%%%%%%%%%%%%%%%%%%%%%%%%%%%%%%%%%%%%%%%%%%%%%%%%%%%%%%%%%%%%%%%%%%%%%%%%%%%%%%%%%%
%%%\vspace{-2mm}  
%\section{The structure  of the proof of Lemma~\ref{main-lemma}}\label{przewodnik}
\section{Structure  of the proof of Lemma~5.7}\label{przewodnik}
%%%\vspace{-2mm}  

The rest of the paper will be devoted to the proof of Lemma~\ref{main-lemma} (restated):

\wazny*

Proof of claim (\ref{black}) -- which will be presented in the end of Section \ref{ulga} -- will  be straightforward once the reader grasps the (slightly complicated) constructions that will 
emerge in the proof of claim (\ref{white}).

For the proof of claim (\ref{white}) we will employ Lemma~\ref{bridge}, showing that if the instance $\langle {\mathcal{S}},{\mathcal{F}} \rangle$ has no proper shading then 
{\em Crocodile} {\bf does have} a winning strategy in Escape($\mathcal{Q}$, $Q_{0}$)  (where $\pair{\mathcal{Q},Q_{0}} = \reduction(\pair{\mathcal{S},\mathcal{F}})$). 
As we remember from Section \ref{game}, in such a game {\em Fugitive} will first choose, as the initial position of the game,  a structure $\database_0= w[a,b]$ 
for some $w\in G(Q_{0})$. Then, in each step, {\em Crocodile} will pick a request in the current structure (current position of the game) $\database$ and {\em Fugitive} will satisfy this request, creating a new (slightly bigger) current $\database$. {\em Fugitive} will win if he will be able to play forever (by which, formally speaking, we mean $\omega^2$ steps, for more details see Section~\ref{game}), or until all requests are satisfied, without satisfying (in the constructed structure) the query $(R(Q_{0}))(a,b)$. 
While talking about the strategy of {\em Fugitive} we will use the words ``must not'' and ``must'' as shorthands for
``or otherwise he will quickly lose the game''. The expression  ``{\em Fugitive} is forced to'' will also have this meaning.

Analysing a two-player game (proving that certain player has a winning strategy) sounds like a complicated task: 
there is this (infinite) alternating tree of positions,  whose structure  somehow needs to be translated into a system of lemmas.
In order to  prune this game tree our plan is first to notice that in his strategy {\em Fugitive} must obey the following principles:\smallskip\\
(I)~The structure $\database_0$ resulting from his initial move must be $(G(w))[a,b]$ for some $w\in Q_{start}$.\\
(II)~He must not allow any green edge with  warm label and any red edge with cold label to appear in $\database$. \\
(III)~He must never allow any path labelled by a word from the language $G(\mathcal{Q}_{bad}) \cup R(\mathcal{Q}_{bad})$ to occur between vertices $a$ and $b$.

Then we will assume that {\em Fugitive's} play indeed follows the three principles and we will present a strategy for 
{\em Crocodile} which will be winning against {\em Fugitive}. From the point of view of {\em Crocodile's} operational objectives this strategy comprises three stages.

In each of these stages {\em Crocodile's} operational goal will be to force {\em Fugitive} to build some specified structure
(where, of course all the specified structures will be superstructures of $\database_0$).
In the first stage {\em Fugitive} will be forced to build a structure called $\PP{1}{\empty}$ (defined in Section~\ref{pm}). In the second stage the specified structures will be called $\PP{m}{\empty}$ and $\PP{m}{\$}$
(each defined in Section~\ref{pm})
and in the third stage  {\em Fugitive} 
will be forced to construct one of the structures $\GG{m}{\empty}$ or $\LL{m}{k}{\empty}$ (defined in Section~\ref{mgrid})

During the three stages of his play {\em Crocodile} will only pick requests from the languages in $\mathcal{Q}_{good}$. These languages, as we said before,
are shade-insensitive, so we can imagine {\em Crocodile} playing in a sort of shade filtering glasses. Of course  {\em Fugitive}, when responding to {\em Crocodile's} requests, will need to commit on the shades of the symbols he will use, but {\em Crocodile's} actions will not depend on these shades. 

The shades will however play their part after the end of the third stage.
Assuming that the original instance of OGTP has no proper shading, we will get that, at this moment, $R(\mathcal{Q}_{bad})(a,b)$ already  holds true in the 
 structure  {\em Fugitive} was forced to construct.  This will end the proof of (\ref{white}).

%%%\vspace{-2mm}  
\section{Principles of the Game} \label{principles-section}
%%%\vspace{-2mm}  

The rules of the game of Escape are such that {\em Fugitive} loses when
he builds a path (from $a$ to $b$) labelled with $w\in R(Q_0)$. So -- when
trying to encode something --  one can think of words in $Q_0$ as of some sort of
forbidden patterns. And thus one can think of  $Q_0$ as of
a tool detecting that {\em Fugitive} is cheating and not really building a
valid computation of the computing device we encode. Having this in
mind the reader can imagine why the
words from languages from the sets 
 ${\mathcal Q}_{bad}$  and  ${\mathcal Q}_{ugly}$,  which clearly are all about suspiciously looking patterns,
 are all in $Q_0$.

But  another  rule of the game is that at the beginning  {\em Fugitive}
picks his initial position ${\mathbb D}_0$ as a path  (from $a$ to $b$)
labelled with some $w\in G(Q_0)$, so it would be nice to think of $Q_0$
as of initial configurations of this computing device. The fact that the same object 
is playing the set of forbidden patterns and, at the same time, the set of initial configurations
is a problem. We solved it by having languages from $\mathcal{Q}_{ugly} \cup \mathcal{Q}_{bad}$ both 
in $\mathcal{Q}$ and in $Q_0$:

%%%\vspace{-2mm}  
\begin{lemma}[\textbf{Principle I}]\label{principle-i}
{\em Fugitive} must choose to start the \textit{Escape} game\\ from \mbox{$\database_{0} = G(w)[a,b]$ for some $w \in Q_{start}$.}
\end{lemma}
%%%\vspace{-2mm}  

Notice that, from the point of view of the  shades-blind {\em Crocodile}  the words in $Q_{start}$ are indistinguishable and thus 
{\em Fugitive} only has one possible choice of $\database_{0}$. 
\begin{proof}
If $\database_{0} = G(q)[a,b]$ for $q \in Q_{0} \setminus Q_{start}$ then $\database_{0} \models \Path{G(L)}{a,b}$ for some $L \in \mathcal{Q}_{ugly} \cup \mathcal{Q}_{bad}$. Then in the next step {\em Crocodile} can pick request $\pair{a,b,\rtgd{G(L)}{R(L)}}$. After {\em Fugitive} satisfies this request, a structure $\database_{1}$ is created such that $\database_{1} \models \Path{R(L)}{a,b}$ and {\em Crocodile} wins.
\end{proof}

From now on we assume that {\em Fugitive} obeys Principle I. This implies that for some $w\in Q_{start}$ structure $G(w)[a,b]$ is contained in all subsequent structures $\database$ at each step of the game.

Now we will formalise the intuition about languages from $\mathcal{Q}_{ugly}$ as forbidden patterns. We start with an observation that simplifies reasoning in the proof of Principle II.

\begin{observation}
\label{obs-princ-2}
For some vertices $u,v$ in the current structure $ \database$ if there is a green (red) edge between them then {\em Crocodile} can force {\em Fugitive} to draw a red (green) edge between $u$ and $v$.
\end{observation}
\begin{proof}
It is possible  due to languages $1-9$ in $\mathcal{Q}_{good}$.
\end{proof}

%NEW
\begin{definition}
A \textit{P2-ready\footnote{Meaning ``ready for Principle II''.}} structure $\database$ is a structure satisfying the following:
\begin{itemize}
    \item $\database_{0}$ is a substructure of $\database$;
    \item there are only two edges incident to $a$: $\pair{a,a'}$ with label $G(\alpha^{C})$ and $\pair{a,a'}$ with label $R(\alpha^{W})$;
    \item all edges labeled with $\alpha^C$ and $\alpha^W$ are between $a$ and $a'$;
    \item there are only two edges incident to $b$: $\pair{b',b}$ with label $G(\omega)$ and  $\pair{b',b}$ with label $R(\omega)$;
    \item all edges labeled with $\omega$ are between $b'$ and $b$;
    \item for each $v \in \database \setminus \{a,b\}$ there is a directed path in $\database$, of length at most $4$, from $a'$ to $v$ and  there is a directed path in $\database$, of length at most $4$, from $v$ to $b'$.
\end{itemize}
\end{definition}

%%%\vspace{-2mm}  
\begin{lemma}[\textbf{Principle II}]\label{principle-ii}
Suppose that, after {\em Fugitive's} move, the current structure $\database$ is a P2-ready structure. Then neither a green edge with label from $\CC{\bullet}{\empty}{W}{\empty}{\empty}$ nor a red edge with label from $\CC{\bullet}{\empty}{C}{\empty}{\empty}$ may appear in $\database$.
\end{lemma}
%%%\vspace{-2mm}  

\begin{proof}
First suppose that there is such a green edge $e = \pair{u,v}$ with label $\CC{\bullet}{\empty}{W}{\empty}{\empty}$ in structure $\database$. Let us denote by $P$ a path from $a'$ to $b'$ through $e$. Observe that if some of the edges of $P$ are red then from Observation~\ref{obs-princ-2} in at most $8$ moves {\em Crocodile} can force {\em Fugitive} to create path $P'$ which goes through the same vertices as $P$ (and also through $e$) but consists only of green edges. Because of this path there is a request generated by $Q_{ugly}^1$ between $a$ and $b$ so in the next step {\em  Crocodile} can force {\em Fugitive} to create a red path connecting $a$ and $b$ labelled with a word from ${Q}_{ugly}^1$, which results in {\em Crocodile's} victory.

The second case is simmilar but uses $Q_{ugly}^2$ instead. 
\end{proof}
%Appendix, Section~\ref{principle-ii-proof}.

%%%\vspace{-2mm}  
\begin{lemma}[\textbf{Principle III}]\label{principle-iii}
{\em Fugitive} must not allow any path labelled with a word from $R(\mathcal{Q}_{bad}) \cup G(\mathcal{Q}_{bad})$ to occur in the current structure $\database$ between vertices $a$ and $b$.
\end{lemma}
%%%\vspace{-2mm}  
\begin{proof}
First consider a case where $\database \models R(\mathcal{Q}_{bad})(a,b)$. Then {\em Fugitive} has already lost as $\mathcal{Q}_{bad} \subset Q_{0}$.

The second case is when $\database \models G(\mathcal{Q}_{bad})(a,b)$ and $\database \not\models R(\mathcal{Q}_{bad})(a,b)$. Then {\em Crocodile} can pick request $\langle a,b,Q_{bad}^{i \rightarrow} \rangle$ (for some $i$) for {\em Fugitive} to satisfy. In both cases after at most one move {\em Fugitive} loses.
\end{proof}
%Appendix, Section~\ref{principle-iii-proof}.

%%%%%%%%%%%%%%%%%%%%%%%%%%%%%%%%%%%%%%%%%%%%%%%%%%%%%%%%%%%%%%%%%%%%%%%%%%%%%%%%

%%%\vspace{-2mm}  
\section{The paths $\PP{m}{\empty}$ and $\PP{m}{\$}$}
%%%\vspace{-2mm}  
\label{pm}

\begin{definition}
(See  Figure~\ref{fig:fig1}, please use a color printer if you can) ${\mathbb P}_{m}$, for $m \in \mathbb{N}_{+}$, is a directed graph $(V,E)$ where $V = \{a, a', b', b\} \cup \{v_{i} : i \in [0,2m] \}$ and the edges $E$ are labelled with symbols from $\Sigma \setminus \Sigma_{0}$ or  with symbols of the form $\CC{l}{o}{t}{\empty}{\empty}$, where -- like before -- $\textbf{l}\in \{A,B\}$, $o\in \{\hor, \ver\}$  and $t\in \{W,C\}$. Each label has to also be either red or green. Notice that there is no $s\in \mathcal S$ here: the labels we now use are sets of symbols from $\bar\Sigma$ like in Notation~\ref{notation-cc}:  we watch the play in {\em Crocodile's} shade filtering glasses.

The edges of $\PP{m}{\empty}$ are as follows:
\begin{itemize}
\item Vertex $a'$ is a successor of $a$ and vertex $b$ is a successor of $b'$.
For each $i \in [2m]$ the successors of $v_i$ are $v_{i+1}$ (if it exists) and $b'$ and the predecessors of $v_i$ are $v_{i-1}$  (if it exists) and $a'$. From each node there are two edges to each of its successors, one red and one green, and there are no other edges.
\item Each \textit{Cold} edge (labelled with a symbol in $\CC{\bullet}{\empty}{C}{\empty}{\empty}$) is green.
\item Each \textit{Warm} edge (labelled with a symbol in $ \CC{\bullet}{\empty}{W}{\empty}{\empty}$) is red.
\item Each edge  $\langle v_{2i}, v_{2i+1}\rangle$ is from $ \CC{A}{\hor}{\empty}{\empty}{\empty}$.
\item Each edge $\langle v_{2i+1},v_{2i+2}\rangle$ is from $\CC{B}{\ver}{\empty}{\empty}{\empty}$.
\item Each edge $\langle a',v_{i} \rangle$ is labelled by either $x^{C}$ or $x^{W}$.
\item Each edge $\langle v_{i},b' \rangle$ is labelled by either $y^{C}$ or $y^{W}$.
\item Edges $\pair{a,a'}$ with label $G(\alpha^{C})$ and $\pair{a,a'}$  with label $R(\alpha^{W})$ are in $E$.
\item Edges $\pair{b',b}$ with label $G(\omega)$ and $\pair{b',b}$ with label $R(\omega)$ are in $E$.
\end{itemize}
\end{definition}

\begin{definition}
$\PP{m}{\$}$ for $m \in \mathbb{N}_{+}$ is $\PP{m}{\empty}$ with two additional edges:
 $\pair{v_{2m},b'} \in E$, with label $G(\$^{C})$,
and  $\pair{v_{2m},b'} \in E$, with label $R(\$^{W})$.
\end{definition}

One may notice\footnote{Not all anonymous reviewers equally appreciated our decision to use the ``exercise'' environment in this paper. In our opinion proving some simple facts themselves, rather than skipping the proofs, can help the readers to develop intuitions needed to understand what is to come in the paper. We discussed this issue with several colleagues and none of them felt that using this environment is arrogant.} that 
$\database_0 $ is a substructure of 
%there is a homomorphism from $G(Q_{start})[a,b]$ to each 
both $\PP{m}{\empty}$ and $\PP{m}{\$}$, and that:

\begin{exercise}\label{noreq}
For any $m$, the only requests generated by $\mathcal{Q}_{good}$ in $\PP{m}{\$}$ are those generated by $Q_{good}^{10}$ and $Q_{good}^{11}$. 
\end{exercise}

\begin{exercise}
Each $\PP{m}{\empty}$ and each $\PP{m}{\$}$ is a \textit{P2-ready} structure.
\end{exercise}

\begin{figure*}
\centering
\includegraphics[width=1.0\textwidth]{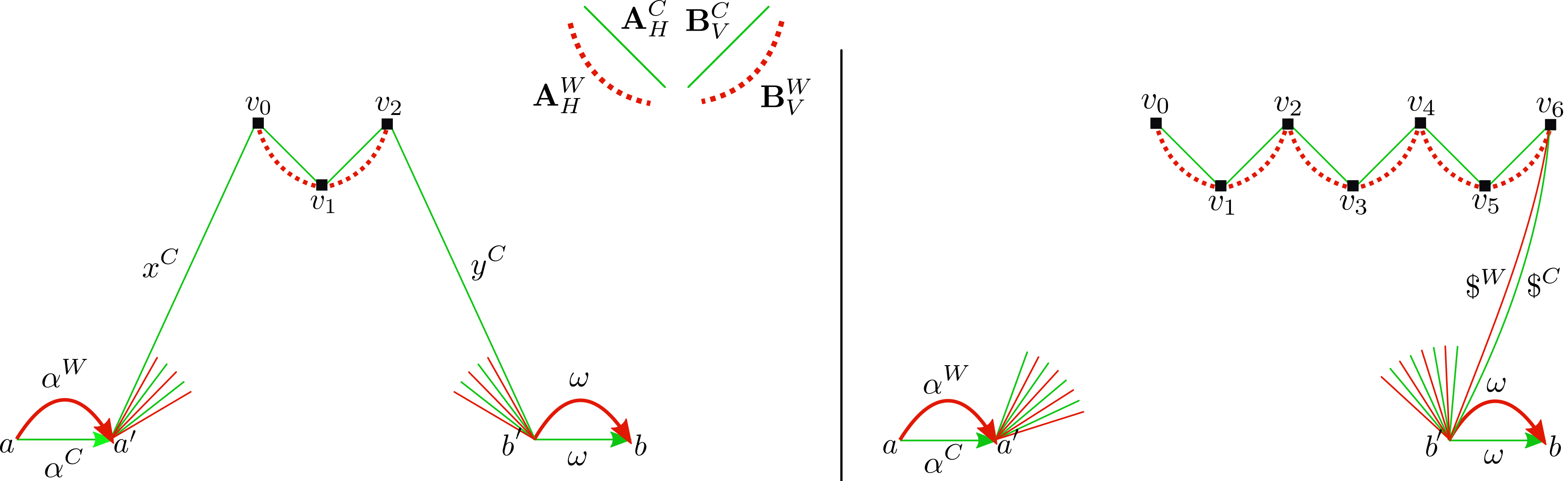}
\caption{\label{fig:fig1} $\PP{1}{\empty}$ (left) and $\PP{3}{\$}$ (right).}
\end{figure*}

%%%%%%%%%%%%%%%%%%%%%%%%%%%%%%%%%%%%%%%%%%%%%%%%%%%%%%%%%%%%%%%%

%%%%%%%%%%%%%%%%%%%%%%%%%%%%%%%%%%%%%%%%%%%%%%%%%%%%%%%%%%%%%%%%%%%%%%%%%%%%%%%%
% Phase I

%%%%%%%%%%%%%%%%%%%%%%%%%%%%%%%%%%%%%%%%%%%%%%%%%%%%%%%%%%%%%%%%%%%%%%%%%%%%%%%%

%%%\vspace{-2mm}  
\section{Stage I}\label{straszna}
%%%\vspace{-2mm}  
Recall that untill the end of Section~\ref{phase2} we watch, analyse {\em Fugitive's} and \textit{Crocodile's} play in shade filtering glasses. And we assume that {\em Fugitive} obeys Principle I and III.

\begin{definition}[\textbf{Crocodile's strategy}]
The sequence of languages $S = (l_{1}, l_{2}, \dots, l_{n})$, for some $n \in \naturals$, defines a strategy for \textit{Crocodile} as follows:
If $S = (l) \concat S'$ (where \concat denotes sequence concatenation) then \textit{Crocodile} demands \textit{Fugitive} to satisfy requests generated by $l$ one by one (in any order) until (it can take infinitely many steps) there are no more requests generated by $l$ in the current structure. Then\footnote{For this ``then'' to make sense we need the total number of moves of the game to be $\omega^2$ rather than $\omega$.}  \textit{Crocodile} proceeds with strategy $S'$.
\end{definition}

Now we define a set of strategies for  {\em Crocodile}. All languages that will appear in these strategies are from $\mathcal{Q}_{good}$ so instead of writing ${Q}_{good}^{i}$ we will just write $i$. Let:
\begin{itemize}
	\item $S_{color} \coloneqq (3,4,5,6,7,8,9)$,
    \item $S_{cycle} \coloneqq (15,14) \concat S_{color} \concat (12,13) \concat S_{color}$,
    \item $S_{start} \coloneqq (1,2) \concat S_{cycle}$.
\end{itemize}

Recall that $\database_{0}$ is \textit{Fugitive's} initial structure (consisting of green edges only), as demanded by Principle I.

%%%\vspace{-2mm}  
\begin{lemma}\label{alpha}
{\em Crocodile's} strategy $(1,2)$ applied to the current structure $\database_{0}$ forces {\em Fugitive} to add $R(\alpha^{W})[a,a']$ and $R(\omega)[b',b]$.
\end{lemma}
%%%\vspace{-2mm}  

\begin{proof}
Consider these languages one by one:

$Q_{good}^1 = \omega$: This language generates only one request $\langle b',b,Q_{good}^{1 \rightarrow} \rangle$ (one because edge $\pair{b',b}$ with label $G(\omega)$ is the only one in  $\database_{0}$ labelled with $\omega$), which has to be satisfied with $R(\omega)[b',b]$ as language $Q_{good}^{1}$ consists of only one word.

$Q_{good}^2 = \alpha^{C} + \alpha^{W}$: There is a green edge labelled with $\alpha^{C}$ in   $\database_{0}$ and thus
this language generates a request $\langle a,a',Q_{good}^{2 \rightarrow} \rangle$ (and no other requests). 
This request can be satisfied by {\em Fugitive} either by adding the edge $R(\alpha^{C})[a,a']$ or by adding the edge $R(\alpha^{W})[a,a']$. Suppose that {\em Fugitive} satisfies this request with $R(\alpha^{C})[a,a']$. Notice that {\em Crocodile} can now require {\em Fugitive} to satisfy requests $Reqs = \{$ $\langle a',v_{0},Q_{good}^{3 \rightarrow} \rangle$, $\langle v_{2},b',Q_{good}^{4 \rightarrow} \rangle$, $\langle v_{0},v_{1},Q_{good}^{9 \rightarrow} \rangle$, $\langle v_{1},v_{2},Q_{good}^{6 \rightarrow} \rangle \}$ which will force {\em Fugitive} to build a red path from $a'$ to $b'$. Each of these request has to be satisfied with a red edge with some label {\em warm} (with the upper index $W$)  or {\em cold} (with $C$).

 Consider what happens if one of these requests is satisfied with a {\em warm} letter. Then we have that $\database \models R(Q_{ugly}^1)(a,b)$ and {\em Fugitive} loses. It means that each request from $Reqs$ must be satisfied with a red edge labelled with a {\em cold} letter. But then notice that $\database \models R(Q_{start})(a,b)$ and {\em Fugitive} also loses.
\end{proof}
%Appendix, Section~\ref{alpha-proof}.

A careful reader could ask here: ``Why did we need to work so hard to prove that the newly added red edge must be 
{\em warm}. Don't we have Principle II which says that red edges must always be {\em warm} and green must be 
{\em cold}?''. But  we cannot use Principle II here: the structure is not \textit{P2-ready} yet. Read
the proof of Principle II again to notice that this red $\alpha^{W}$ between $a$ and $a'$ is crucial there.  And this is what Stage I is all about: it is here where {\em  Crocodile} forces {\em Fugitive} to construct a  structure which is \textit{P2-ready}. From now on all the current structures will be 
\textit{P2-ready} and  {\em Fugitive} will indeed be a slave of Principle II.

The following Lemma explains the role of $S_{color}$ and is a first cousin of Observation~\ref{obs-princ-2}:

%%%\vspace{-2mm}  
\begin{lemma}[$\mathbf{S_{color}}$]\label{color}
Strategy $S_{color}$ applied to a P2-ready $\database$ forces {\em Fugitive} to create a P2-ready $\database'$ such that:
\begin{itemize}
\item Sets of vertices of $\database$ and $\database'$ are equal.
\item There are no requests generated by $Q_{good}^{1-9}$ in $\database'$, which means that each edge has its counterpart (incident to the same vertices) of the opposite color and temperature.  
\end{itemize}
\end{lemma}

\begin{proof}
The proof is an easy consequence of Principle II and the fact that all words from $Q_{good}^{1-9}$ have length one (which means that when satisfying  the requests {\em Fugitive} only creates new edges, but no new vertices are added) and that these languages contain all symbols from $\Sigma$. 
\end{proof}

\begin{lemma}\label{P1}
Strategy $S_{start}$ applied to $\database_0$ forces {\em Fugitive} to build $\PP{1}{\empty}$.
\end{lemma}
%%%\vspace{-2mm}  

\begin{proof}
Consider languages from $S_{start}$ one by one:
\begin{itemize}
\item $1 = \omega$: By Lemma~\ref{alpha} this language forces {\em Fugitive} to add $R(\omega)[b',b]$.
\item $2 = \alpha^{C} + \alpha^{W}$: By Lemma~\ref{alpha} this language forces {\em Fugitive} to add $R(\alpha^{W})[a,a']$.
\item $15 = y^{W} + \$^{C} + \CC{A}{\hor}{C}{\empty}{\empty}\CC{B}{\ver}{C}{\empty}{\empty}y^{C} + \CC{B}{\ver}{C}{\empty}{\empty}y^{C}$: This language generates two requests: $\langle v_{0},b',Q_{good}^{15 \rightarrow} \rangle$ and $\langle v_{1},b',Q_{good}^{15 \rightarrow} \rangle$ since neither $y^W$ nor $\$^C$ occurs in the current structure. The first request has to be satisfied with $R(y^{W})[v_{0},b']$ by Principle II and the second request has to be satisfied $R(y^{W})[v_{1},b']$ by Principle II. We can use here Principle II since after strategy $(1,2)$ was applied the structure was \textit{P2-ready}.
\item $14 = x^{W} + x^{C} + x^{C}\CC{A}{\hor}{C}{\empty}{\empty}\CC{B}{\ver}{C}{\empty}{\empty}$: This language generates only two requests $\langle a',v_{2},Q_{good}^{14 \rightarrow} \rangle$ and $\langle a',v_{0},Q_{good}^{14 \rightarrow} \rangle$. The first request has to be satisfied with $R(x^{W})[a',v_{0}]$ and the second with $R(x^{W})[a',v_{2}]$, both due to Principle II.
\end{itemize}

Now {\em Crocodile} uses strategy $S_{color}$ to add missing edges of opposite colors (and, by Principle II, of opposite temperatures).

\begin{itemize}
\item $12 = x^{C} \left(\CC{A}{\hor}{C}{\empty}{\empty} + \CC{B}{\hor}{C}{\empty}{\empty} + \CC{A}{\ver}{C}{\empty}{\empty} + \CC{B}{\ver}{C}{\empty}{\empty} \right) + x^{C} + x^{W}$: This language generates only one request: $\langle a',v_{1},Q_{good}^{12 \rightarrow} \rangle$. It is because there are no requests generated by neither $x^C$ nor $x^W$ in $Q_{good}^{12}$ by Lemma~\ref{color}.  There are also no other requests generated by $x^{C} \left(\CC{A}{\hor}{C}{\empty}{\empty} + \CC{B}{\hor}{C}{\empty}{\empty} + \CC{A}{\ver}{C}{\empty}{\empty} + \CC{B}{\ver}{C}{\empty}{\empty} \right)$ in $Q_{good}^{12}$ as the only path labeled with a word from this language is $a' \rightarrow v_0 \rightarrow v_1$. $\langle a',v_{1},Q_{good}^{12 \rightarrow} \rangle$ has to be satisfied with $R(x^{W})[a',v_{1}]$ by Principle II.

\item $13 = \left(\CC{A}{\hor}{C}{\empty}{\empty} + \CC{B}{\hor}{C}{\empty}{\empty} + \CC{A}{\ver}{C}{\empty}{\empty} + \CC{B}{\ver}{C}{\empty}{\empty} \right)y^{C} + y^{C} + y^{W}$: This language doesn't generate any requests. %This language generates one request $\langle v_{1},b',Q_{good}^{13 \rightarrow} \rangle$. It has to be satisfied with $R(y^{W})[v_{1},b']$ by Principle II.
\end{itemize}

Finally {\em Crocodile} uses strategy $S_{color}$ to add one missing edge $\pair{a',v_{1}}$ with label $G(x^{C})$ to build $\PP{1}{\empty}$
\end{proof}
%Appendix, Section~\ref{P1-proof}.

%%%\vspace{-2mm}  
\section{Stage II}
%%%\vspace{-2mm}  
\label{stage2}
%NEW
Note that from now on {\em Fugitive} must obey all Principles.

Now we imagine that $\PP{1}{\empty}$ has already been created and we proceed with the analysis to the later stage of the Escape game where either $\PP{m+1}{\empty}$ or $\PP{k}{\$}$ for some $k \leq m$ will be created.

Let us define $\set{S_k}$ inductively for $k \in \mathbb{N}_{+}$ in the following fashion:
\begin{itemize}
\item $S_{1} \coloneqq S_{start}$,
\item $S_{k} \coloneqq S_{k-1} \concat S_{cycle}$ for $k > 1$.
\end{itemize}

%%%\vspace{-2mm}  
\begin{lemma}\label{Pk}
For all $m \in \mathbb{N}_{+}$ strategy $S_{m}$ applied to $\database_0$ forces {\em Fugitive} to build a structure isomorphic, depending on his choice, either $\PP{m+1}{\empty}$ or $\PP{k}{\$}$ for some $k \leq m$.
\end{lemma}
%%%\vspace{-2mm}  

\begin{proof}
Notice that by Lemma~\ref{P1}, this is already proved for $m=1$. Now assume that {\em Crocodile}, using strategy $S_{m-1}$, forced {\em Fugitive} to build $\PP{m}{\empty}$ or $\PP{k}{\$}$, for some $k \leq m-1$. If {\em Fugitive} already built $\PP{k}{\$}$ as the result of {\em Crocodile's} strategy $S_{m-1}$ then we are done, by noticing that
the last $S_{cycle}$ will not change the current structure any more -- this is because, due to Exercise \ref{noreq} there are 
no requests from languages $Q_{good}^{1-9}$ and $Q_{good}^{12-15}$ in the current structure at this point. 

So we only need to consider the case where $\PP{m}{\empty}$ was built. Now {\em Crocodile} uses strategy $S_{cycle}$ to force {\em Fugitive} to build $\PP{m+1}{\empty}$ or $\PP{m}{\$}$. 
Consider languages from $S_{cycle}$ one by one:
\begin{itemize}
\item $15 = y^{W} + \$^{C} + \CC{A}{\hor}{C}{\empty}{\empty}\CC{B}{\ver}{C}{\empty}{\empty}y^{C} + \CC{B}{\ver}{C}{\empty}{\empty}y^{C}$. The only request generated by this language is $\langle v_{2m},b',Q_{good}^{15 \leftarrow} \rangle$,
resulting from the red edge labelled with $y^{W}$ connecting $v_{2m}$ and $b'$.

This is since:
\begin{itemize}
    \item there is no $\$^{C}$ anywhere in the current structure,
    \item for each $k<m$
there are already both a red edge labelled with $y^{W}$ from $v_{2k}$ to $b'$ and 
a green path labelled with $\CC{A}{\hor}{C}{\empty}{\empty}\CC{B}{\ver}{C}{\empty}{\empty}y^{C}$ between these vertices,

    \item for each $k<m$
there are already both a red edge labelled with $y^{W}$ from $v_{2k+1}$ to $b'$ and 
a green path labelled with $\CC{B}{\ver}{C}{\empty}{\empty}y^{C}$ between these vertices.
\end{itemize}

This only request can possibly be satisfied in three different ways (it follows from Principle II): either by $G(\CC{A}{\hor}{C}{\empty}{\empty}\CC{B}{\ver}{C}{\empty}{\empty}y^{C})[v_{2m},b']$ or by $G(\CC{B}{\ver}{C}{\empty}{\empty}y^{C})[v_{2m},b']$ or by $G(\$^{C})[v_{2m},b']$. First notice that this request cannot be satisfied with $G(\CC{B}{\ver}{C}{\empty}{\empty}y^{C})[v_{2m},b']$ because that would result in creating a green path labeled by $\alpha^C x^C \CC{B}{\ver}{C}{\empty}{\empty}\CC{B}{\ver}{C}{\empty}{\empty}y^C\omega$ connecting $a$ and $b$. Then {\em Crocodile} could pick request $\langle a,b,Q_{ugly}^{3 \rightarrow} \rangle$ for {\em Fugitive} to satisfy. After {\em Fugitive} satisfies that he will lose. The case when this request is satisfied with $G(\$^{C})[v_{2m},b']$ will be considered in the last paragraph of the proof. So now we assume that this request is satisfied with $G(\CC{A}{\hor}{C}{\empty}{\empty}\CC{B}{\ver}{C}{\empty}{\empty}y^{C})[v_{2m},b']$. Let us name the two new vertices $v_{2m+1}$ and $v_{2m+2}$.

\item $14 = x^{W} + x^{C} + x^{C}\CC{A}{\hor}{C}{\empty}{\empty}\CC{B}{\ver}{C}{\empty}{\empty}$: the only request generated by this language is $\langle a',v_{2m+2},Q_{good}^{14 \rightarrow}\rangle$ resulting
from the (partially) newly created green path from $a'$ to $v_{2m+2}$, via $v_{2m}$ and $v_{2m+1}$,
labelled with $x^{C}\CC{A}{\hor}{C}{\empty}{\empty}\CC{B}{\ver}{C}{\empty}{\empty}y^{C}$.

This request has to be satisfied with $R(x^{W})[a',v_{2m+2}]$ due to Principle II.

\end{itemize}

Now  {\em Crocodile} uses strategy $S_{color}$ to add missing edges of opposite colors.

\begin{itemize}
\item $12 = x^{C} \left(\CC{A}{\hor}{C}{\empty}{\empty} + \CC{B}{\hor}{C}{\empty}{\empty} + \CC{A}{\ver}{C}{\empty}{\empty} + \CC{B}{\ver}{C}{\empty}{\empty} \right) + x^{C} + x^{W}$: This language generates one request: $\langle a',v_{2m+1},Q_{good}^{12 \rightarrow} \rangle$. It has to be satisfied with $R(x^{W})[a',2m+1]$ by Principle II.

\item $13 = \left(\CC{A}{\hor}{C}{\empty}{\empty} + \CC{B}{\hor}{C}{\empty}{\empty} + \CC{A}{\ver}{C}{\empty}{\empty} + \CC{B}{\ver}{C}{\empty}{\empty} \right)y^{C} + y^{C} + y^{W}$: This language generates one request $\langle v_{2m+1},b',Q_{good}^{13 \rightarrow} \rangle$. It has to be satisfied with $R(y^{W})[2m+1,b']$ by Principle II.
\end{itemize}
Now {\em Crocodile} uses strategy $S_{color}$ (as $S_{cycle} = (15,14) \concat S_{color} \concat (12,13) \concat S_{color}$). We apply Lemma~\ref{color} to conclude that {\em Fugitive} is forced to build $\PP{m+1}{\empty}$, as what is left to create $\PP{m+1}{\empty}$ is to only add some edges of opposite colors and temperatures. 

Notice that during play, after application of each language in {\em Crocodile's} strategy, each of the constructed structures is \textit{P2-ready}, as distances from $a'$ and to $b'$ are smaller than $4$.

Now we finally consider the case where {\em Fugitive} satisfied the request generated by language $15$ with $G(\$^{C})[v_{m},b']$. Notice that the only request generated by the remaining languages from $S_{cycle}$ is: $\langle v_{2m},b',Q_{good}^{5 \rightarrow} \rangle$, which will be satisfied by $R(\$^{W})[v_{2m},b']$ and the resulting structure will be isomorphic to $\PP{m}{\$}$. This ends the proof of Lemma~\ref{Pk}.
\end{proof}

%%%%%%%%%%%%%%%%%%%%%%%%%%%%%%%%%%%%%%%%%%%%%%%%%%%%%%%%%%%%%%%%%%%%%%%%%%%%%%%%
% grid

%%%%%%%%%%%%%%%%%%%%%%%%%%%%%%%%%%%%%%%%%%%%%%%%%%%%%%%%%%%%%%%%%%%%%%%%%%%%%%%%

%%%%%%%%%%%%%%%%%%%%%%%%%%%%%%%%%%%%%%%%%%%%%%%%%%%%%%%%%%%%%%%%%%%%%%%%%%%%%%%%

%%%\vspace{-2mm}  
\section{The grids $\GG{m}{\empty}$ and partial grids $\LL{m}{k}{\empty}$}\label{mgrid}
%%%\vspace{-2mm}  
%\section{The grids $\GG{m}{\empty}, \GG{m}{\$}, \LL{m}{k}{\empty}$ and $\LL{m}{k}{\$}$}\label{mgrid}

\begin{definition}
$\GG{m}{\empty}$, for $m \in \mathbb{N}_{+}$, is a directed graph $(V,E)$ where:

$V = \{a, a', b', b\} \cup \{v_{i,j} : i,j \in [0,m] \}$ and the edges $E$ are labelled (as in $\PP{m}{\empty}$) with $\Sigma \setminus \Sigma_{0}$ or one of the symbols of the form $\CC{l}{o}{t}{\empty}{\empty}$, which means that the shade filtering glasses are still on.

The edges of $\GG{m}{\empty}$ are as follows:
\begin{itemize}
\item Vertex $a'$ is a successor of $a$, $b$ is a successor of $b'$. 
All $v_{i,j}$ are successors of $a'$  and the successors of each $v_{i,j}$ are $v_{i+1,j}, v_{i,j+1}$ (when they exist) and $b'$. From each node there are two edges to each of its successors, one red and one green. There are no other edges.
\item Each {\em cold} edge, labelled with a symbol in $\CC{\bullet}{\empty}{C}{\empty}{\empty}$, is green.
\item Each {\em warm} edge, labelled with a symbol in $ \CC{\bullet}{\empty}{W}{\empty}{\empty}$, is red.
\item Each edge  $\langle v_{i,j}, v_{i+1,j}\rangle$ is horizontal -- its label is from $ \CC{\bullet}{\hor}{\empty}{\empty}{\empty}$.
\item Each edge $\langle v_{i,j},v_{i,j+1}\rangle$ is vertical -- its label is from $\CC{\bullet}{\ver}{\empty}{\empty}{\empty}$.
\item The label of each edge leaving $v_{i,j}$, with $i+j$ even, is from $ \CC{A}{\empty}{\empty}{\empty}{\empty}$, the label of each edge leaving $v_{i,j}$, with $i+j$ odd, is from $ \CC{B}{\empty}{\empty}{\empty}{\empty}$.
\item Each edge $\langle a',v_{i} \rangle$ is labeled by either $x^{C}$ or $x^{W}$.
\item Each edge $\langle v_{i},b' \rangle$ is labeled by either $y^{C}$ or $y^{W}$.
\item Edges $\pair{a,a'}$ with label $G(\alpha^{C})$ and $\pair{a,a'}$ with label $R(\alpha^{W})$ are in $E$.
\item Edges $\pair{b',b}$ with label $G(\omega)$ and $\pair{b',b}$ with label $R(\omega)$ are in $E$.
\end{itemize}
\end{definition}

%%%\vspace{-2mm}
\begin{definition}
Let $\LL{m}{k}{\empty} = (V',E')$, for $m, k \in \mathbb{N}_{+}$ where $ k \leq m$, is a subgraph of $\GG{m}{\empty} = (V,E)$ induced by the set $V'\hspace{-0mm} \subseteq \hspace{-0mm} V$ of vertices   defined as $V' \hspace{-0mm}= \{a,a',b',b\} \cup \{v_{i,j} : i,j \in [0,m]; \\ |~i~-~j~| \leq~k~\}$.
\end{definition}
%%%\vspace{-2mm}

\begin{definition}
Let $\GG{m}{\$}$ for $m \in \mathbb{N}_{+}$ is $\GG{m}{\empty}$ with two edges added:
$\pair{v_{m,m},b'}$ with label $G(\$^{C})$ and $\pair{v_{m,m},b'}$ with label $R(\$^{W})$.
\end{definition}
%%%\vspace{-2mm}

\begin{definition}
Let $\LL{m}{k}{\$}$ for $m \in \mathbb{N}_{+}, k \in \mathbb{N}_{+} \cup \{0\}, k \leq m$ is $\LL{m}{k}{\empty}$ with two edges added:
$\pair{v_{m,m},b'}$ with label $G(\$^{C})$ and $\pair{v_{m,m},b'}$ with label $R(\$^{W})$.
\end{definition}
%%%\vspace{-2mm}

\begin{fact}\label{iso}
For all $m$:
$\LL{m}{m}{\empty}$ is equal to $\GG{m}{\empty}$ and $\LL{m}{m}{\$}$ is equal to $\GG{m}{\$}$.
\end{fact}
%%%\vspace{-2mm}

\begin{exercise}\label{noreqG}
Languages from $\mathcal{Q}_{good}$ or $\mathcal{Q}_{ugly}$ do not generate requests in any $\GG{m}{\$}$.
\end{exercise}
%%%\vspace{-2mm}

\begin{figure*}
\centering
\includegraphics[width=0.6\textwidth]{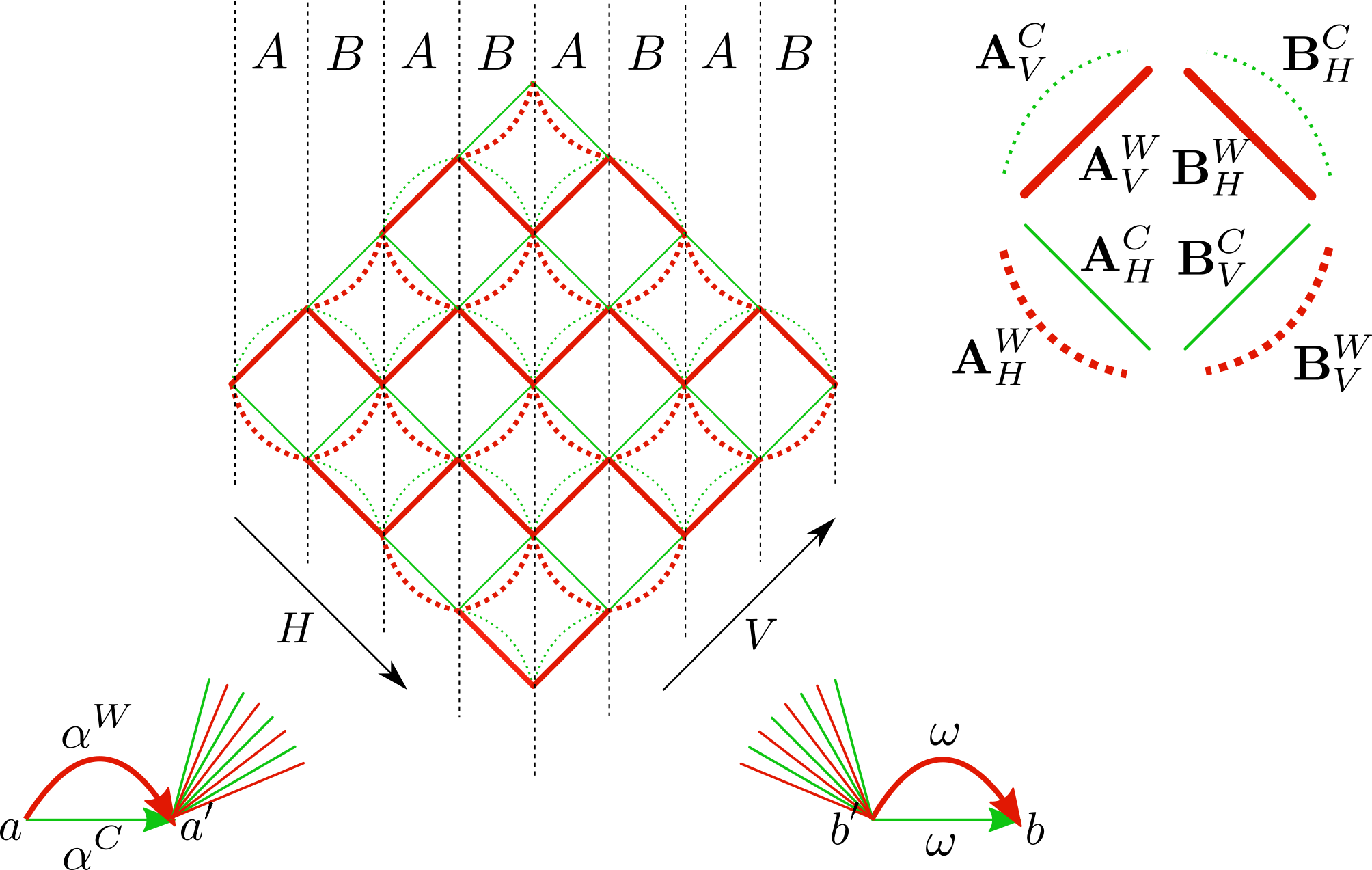}
\caption{\label{fig:fig3} $\GG{4}{\empty}$ (left). Smaller picture in the top-right corner explains how  different line styles on the main picture map to $\Sigma_{0}$ (please use a color printer if you can).}
\end{figure*}

\begin{figure*}
\centering
\includegraphics[width=\textwidth]{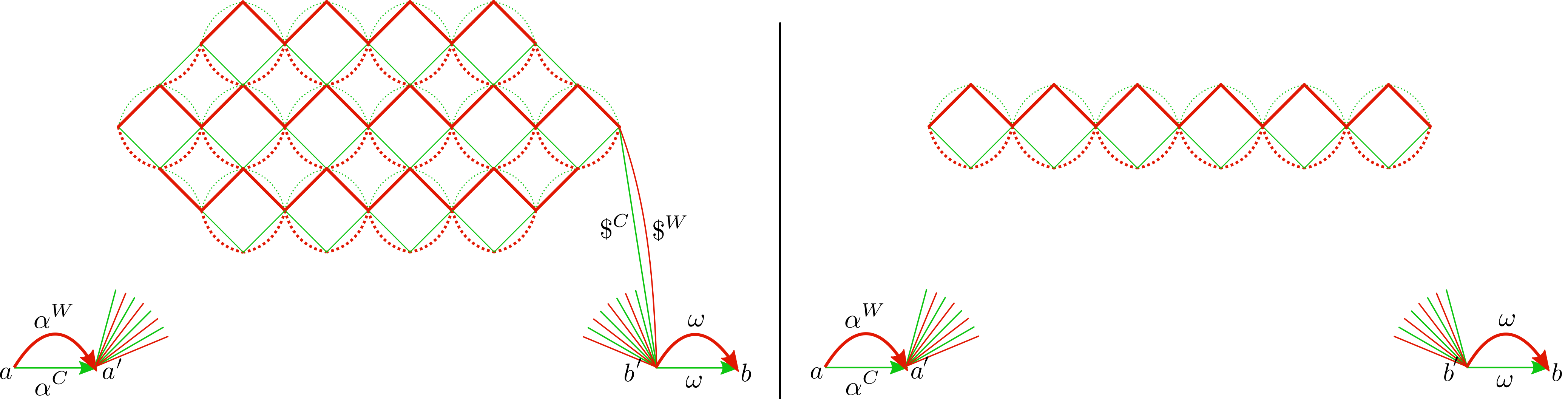}
\caption{\label{fig:fig4} $\LL{6}{3}{\$}$ (left) and $\LL{6}{1}{\empty}$ (right).}
\end{figure*}

%%%%%%%%%%%%%%%%%%%%%%%%%%%%%%%%%%%%%%%%%%%%%%%%%%%%%%%%%%%%%%%%%%%%%%%%%%%%%%%%
% Phase II

%%%%%%%%%%%%%%%%%%%%%%%%%%%%%%%%%%%%%%%%%%%%%%%%%%%%%%%%%%%%%%%%%%%%%%%%%%%%%%%%

\section{Stage III}
%%%\vspace{-2mm}  
\label{phase2}

Now we imagine that either $\PP{m+1}{\empty}$ or $\PP{k}{\$}$ for some $k \leq m$ was created
as the current position in a play of the game of Escape and we proceed with the analysis to the later stage of the play, where either $\GG{m+1}{\empty}$ or $\GG{k}{\$}$ will be created.

%%%\vspace{-2mm}  
\begin{lemma}\label{grid}
For any $m \in \mathbb{N}_{+}$ \textit{Crocodile} can force  \textit{Fugitive} to build a structure isomorphic, depending on \textit{Fugitive's} choice, 
to either $\GG{m+1}{\empty}$ or to $\GG{k}{\$}$ for some $k \leq m$.
\end{lemma}
%%%\vspace{-2mm}  

Notice that by Exercise~\ref{iso}, in order  to prove Lemma~\ref{grid} it is enough to prove that for any $m \in \mathbb{N}_{+}$ \textit{Crocodile} can force \textit{Fugitive} to build a structure isomorphic to either $\LL{m+1}{m+1}{\empty}$ or to $\LL{k}{k}{\$}$ for some $k \leq m$.
% as $\LL{m+1}{m+1}{\empty} \cong \GG{m+1}{\empty}$ and $\LL{k}{k}{\$} \cong \GG{k}{\$}$.

As we said, we assume  that  {\em Crocodile} already forced {\em Fugitive} to build a structure isomorphic to either $\PP{m+1}{\empty}$ or to $\PP{k}{\$}$ for some $k \leq m$.
Rename each $v_i$ in this  $\PP{m+1}{\empty}$ (or $\PP{k}{\$}$) as $v_{i,i}$.
If the structure which was built is $\PP{m+1}{\empty}$  we will show a strategy  leading
to $\LL{m+1}{m+1}{\empty}$ and when $\PP{k}{\$}$ was built, we will show a strategy  leading to $\LL{k}{k}{\$}$. 

Now we define a sequence of strategies $S_{layer}^{k}$, which, similarly to strategies for building $\PP{\bullet}{\bullet}$ consist only of languages from $\mathcal{Q}_{good}$, so instead of writing ${Q}_{good}^{i}$ we will just write $i$.

\vspace{0mm}
Let:
\begin{itemize}
\item $S^{odd} \coloneqq (11) \concat S_{color} \concat (12,13) \concat S_{color}$,
\item $S^{even} \coloneqq (10) \concat S_{color} \concat (12,13) \concat S_{color}$,
\item $
    S_{layer}^{k} \coloneqq
    \begin{cases}
      [\hspace{0.5mm}],  & \text{if}\ k = 0 \\
      S_{layer}^{k-1} \concat S^{odd} & \text{if}\ k \ odd \\
      S_{layer}^{k-1} \concat S^{even}  & \text{otherwise}
    \end{cases}
  	  $
\end{itemize}

\begin{lemma}\label{L1}
For all $k \in \mathbb{N}$ strategy $S_{layer}^{1}$ applied to the current structure $\PP{k}{\$}$ forces {\em Fugitive} to build $\LL{k}{1}{\$}$.
\end{lemma}

\begin{proof}
Assume the current structure is  $\PP{k}{\$}$. Consider languages from $S_{layer}^{1}$:
\begin{itemize}
\item $11 = 
\CC{A}{\hor}{C}{\empty}{\empty}\CC{B}{\ver}{C}{\empty}{\empty} + \CC{A}{\ver}{W}{\empty}{\empty}\CC{B}{\hor}{W}{\empty}{\empty}$: 
This language generates one request of the form  $\langle v_{i},v_{i+2},Q_{good}^{11 \rightarrow} \rangle $ for every  $ i \in [0,2k-2] \}$. Each of these requests results from a green path labeled with $G(\CC{A}{\hor}{C}{\empty}{\empty}\CC{B}{\ver}{C}{\empty}{\empty})$ connecting $v_i$ and $v_{i+2}$. 

Notice that there are no requests generated by $Q_{good}^{11 \leftarrow}$. It is because neither $ \CC{A}{\ver}{W}{\empty}{\empty}$ nor  $\CC{B}{\hor}{W}{\empty}{\empty}$ occurs in $\PP{k}{\$}$.  

All generated requests have to be satisfied with $R(\CC{A}{\ver}{W}{\empty}{\empty}\CC{B}{\hor}{W}{\empty}{\empty})$ by Principle II. Notice that when satisfying each request a new vertex is created.

\item $S_{color} = (3,4,5,6,7,8,9)$: This sequence of languages adds missing green edges $G\CC{A}{\hor}{C}{\empty}{\empty}$ and $G\CC{B}{\ver}{C}{\empty}{\empty}$ to the edges $R\CC{A}{\hor}{W}{\empty}{\empty}$ and $R\CC{B}{\ver}{W}{\empty}{\empty}$ created by language $11$.

\item $12 = x^{C} \left(\CC{A}{\hor}{C}{\empty}{\empty} + \CC{B}{\hor}{C}{\empty}{\empty} + \CC{A}{\ver}{C}{\empty}{\empty} + \CC{B}{\ver}{C}{\empty}{\empty} \right) + x^{C} + x^{W}$:  This language generates requests of the form $\{ \langle a',u,Q_{good}^{12 \rightarrow} \rangle$ for all new vertices $u$ created by language $11$. Each of these requests results from a green path labeled with $x^{C} \left(\CC{A}{\hor}{C}{\empty}{\empty} + \CC{B}{\hor}{C}{\empty}{\empty} + \CC{A}{\ver}{C}{\empty}{\empty} + \CC{B}{\ver}{C}{\empty}{\empty} \right)$ connecting $a'$ and $u$, for some vertex $u$ created by language $11$.

Notice that there are no other requests generated since by Lemma~\ref{color} after applying strategy $S_{color}$ each edge labeled with $G(x^C)$ has its counterpart labeled with $R(x^W)$. 

All generated requests have to be satisfied with $R(x^{W})$ by Principle II. 

\item $13 =  \left(\CC{A}{\hor}{C}{\empty}{\empty} + \CC{B}{\hor}{C}{\empty}{\empty} + \CC{A}{\ver}{C}{\empty}{\empty} + \CC{B}{\ver}{C}{\empty}{\empty} \right)y^{C} + y^{C} + y^{W}$: This language generates requests of the form $\{ \langle u,b',Q_{good}^{13 \rightarrow} \rangle$ for all new vertices $u$ created by language $11$. Each of these requests results from a green path labeled with $(\left \CC{A}{\hor}{C}{\empty}{\empty} + \CC{B}{\hor}{C}{\empty}{\empty} + \CC{A}{\ver}{C}{\empty}{\empty} + \CC{B}{\ver}{C}{\empty}{\empty} \right)y^{C}$ connecting $u$ and $b'$, for some vertex $u$ created by language $11$.

Notice that there are no other requests generated since by Lemma~\ref{color} after applying strategy $S_{color}$ each edge labeled with $G(y^C)$ has its counterpart labeled with $R(y^W)$. 

All these requests have to be satisfied with $R(y^{W})$ by Principle II. 

\item $S_{color} = (3,4,5,6,7,8,9)$: This sequence of languages adds missing green edges $G(x^{C})$ and $G(y^{C})$ to edges added by languages $12$ and $13$.
\end{itemize}
\end{proof}

%%%\vspace{-2mm}  
\begin{lemma}\label{layer}
For all $k,m \in \mathbb{N}, k < m$ strategy $S^{odd}$ (for $k+1$ odd) and $S^{even}$ (for $k+1$ even) applied to $\LL{m}{k}{\$}$ forces {\em Fugitive} to build $\LL{m}{k+1}{\$}$.
\end{lemma}
%%%\vspace{-2mm}  
\begin{proof}
Assume the Escape game starts from $\LL{m}{k}{\$}$ for odd $k < m$. The proof for the case where $k$ is even is analogous. Consider languages from $S^{even}$:

\begin{itemize}
\item $\mathbf{10} = \CC{B}{\hor}{W}{\empty}{\empty}\CC{A}{\ver}{W}{\empty}{\empty} + \CC{B}{\ver}{C}{\empty}{\empty}\CC{A}{\hor}{C}{\empty}{\empty}$: generates exactly \\$\{ \langle v_{i,j},v_{i+1,j+1},Q_{good}^{10 \rightarrow} \rangle | i-j = k, i,j \in [0,m-1] \} \cup$ $\{ \langle v_{i,j},v_{i+1,j+1},Q_{good}^{10 \leftarrow} \rangle | i-j = k, i,j \in [0,m-1] \}$. All requests in the first group result from paths labeled with $G(\CC{B}{\ver}{C}{\empty}{\empty}\CC{A}{\hor}{C}{\empty}{\empty})$ and all requests in the second group result from paths labeled with $R(\CC{B}{\hor}{W}{\empty}{\empty}\CC{A}{\ver}{W}{\empty}{\empty})$.

All requests in the first group have to be satisfied with $R(\CC{B}{\hor}{W}{\empty}{\empty}\CC{A}{\ver}{W}{\empty}{\empty})$ (name the new vertices $v_{i+1,j}$) and all requests in the second group have to be satisfied with $G(\CC{B}{\ver}{C}{\empty}{\empty}\CC{A}{\hor}{C}{\empty}{\empty})$ (name the new vertices $v_{i,j+1}$). All happens by Principle II.

\item $S_{color}$: adds missing edges of opposite colors incident to newly created vertices by language $11$.

\item $\mathbf{12} = x^{C} \left(\CC{A}{\hor}{C}{\empty}{\empty} + \CC{B}{\hor}{C}{\empty}{\empty} + \CC{A}{\ver}{C}{\empty}{\empty} + \CC{B}{\ver}{C}{\empty}{\empty} \right) + x^{C} + x^{W}$: generates exactly $\{ \langle a',v_{i,j},Q_{good}^{12 \rightarrow} \rangle | i-j = k+1, i,j \in [0,m] \} \cup \{ \langle a',v_{i,j},Q_{good}^{12 \rightarrow} \rangle | j-i = k+1, i,j \in [0,m] \}$. Each of these requests results from a green path labeled with $x^{C} \left(\CC{A}{\hor}{C}{\empty}{\empty} + \CC{B}{\hor}{C}{\empty}{\empty} + \CC{A}{\ver}{C}{\empty}{\empty} + \CC{B}{\ver}{C}{\empty}{\empty} \right)$ connecting $a'$ and $u$, for some vertex $u$ created by language $10$.

Notice that there are no other requests generated since by Lemma~\ref{color} after applying strategy $S_{color}$ each edge labeled with $G(x^C)$ has its counterpart labeled with $R(x^W)$

All generated requests have to be satisfied with $R(x^{W})$ by Principle II. 

\item $\mathbf{13} = \left(\CC{A}{\hor}{C}{\empty}{\empty} + \CC{B}{\hor}{C}{\empty}{\empty} + \CC{A}{\ver}{C}{\empty}{\empty} + \CC{B}{\ver}{C}{\empty}{\empty} \right)y^{C} + y^{C} + y^{W}$: generates exactly $\{ \langle v_{i,j},b',Q_{good}^{13 \rightarrow} \rangle | i-j = k+1, i,j \in [0,m] \} \cup \{ \langle v_{i,j},b',Q_{good}^{13 \rightarrow} \rangle | j-i = k+1, i,j \in [0,m] \}$. Each of these requests results from a green path labeled with $(\left \CC{A}{\hor}{C}{\empty}{\empty} + \CC{B}{\hor}{C}{\empty}{\empty} + \CC{A}{\ver}{C}{\empty}{\empty} + \CC{B}{\ver}{C}{\empty}{\empty} \right)y^{C}$ connecting $u$ and $b'$, for some vertex $u$ created by language $10$.

Notice that there are no other requests generated since by Lemma~\ref{color} after applying strategy $S_{color}$ each edge labeled with $G(y^C)$ has its counterpart labeled with $R(y^W)$

All generated requests have to be satisfied with $R(y^{W})$ by Principle II. 

\item $S_{color}$: adds edges with labels $G(x^{C})$ and $G(y^{C})$ to edges added by languages $12$ and $13$.
\end{itemize}

\end{proof}

%%%\vspace{-2mm}  
\begin{lemma}\label{nodolar}
For all $k,m \in \mathbb{N}, k < m$ strategy $S_{layer}^{1}$ applied to $\PP{k}{\empty}$ forces {\em Fugitive} to build $\LL{k}{1}{\empty}$, strategy $S^{odd}$ (for $k+1$ odd) and $S^{even}$ (for $k+1$ even) applied to $\LL{m}{k}{\empty}$ forces {\em Fugitive} to build $\LL{m}{k+1}{\empty}$.
\end{lemma}
%%%\vspace{-2mm}  
\begin{proof}
Similar analysis to that in Lemma~\ref{L1} and Lemma~\ref{layer} can be applied here. Structures $\PP{m}{\empty}$ and $\PP{m}{\$}$ differ by only two edges labeled with $R(\$^{W})$ and $G(\$^{C})$. Letters $\$^{C}$ and $\$^{W}$ occur only in languages $Q_{good}^{5}$ and $Q_{good}^{15}$, these languages didn't generate any request in the process of building $\LL{m}{k+1}{\$}$ from $\LL{m}{k}{\$}$ in the proof of Lemma~\ref{layer} and building $\LL{m}{1}{\$}$ from $\PP{m}{\$}$ in the proof of Lemma~\ref{L1}.
\end{proof}

%%%%%%%%%%%%%%%%%%%%%%%%%%%%%%%%%%%%%%%%%%%%%%%%%%%%%%%%%%%%%%%%%%%%%%%%%%%%%%%%%%%%%%%%%%%%%%%%%%%%%%%%
%%%%%%%%%%%%%%%%%%%%%%%%%%%%%%%%%%%%%%%%%%%%%%%%%%%%%%%%%%%%%%%%%%%%%%%%%%%%%%%%%%%%%%%%%%%%%%%%%%%%%%%%
%%%%%%%%%%%%%%%%%%%%%%%%%%%%%%%%%%%%%%%%%%%%%%%%%%%%%%%%%%%%%%%%%%%%%%%%%%%%%%%%%%%%%%%%%%%%%%%%%%%%%%%%

%%%\vspace{-2mm}  
\begin{lemma}\label{Lmm}
For all $m \in \mathbb{N}$ strategy $S_{layer}^{m}$ forces {\em Fugitive} to build $\LL{m}{m}{\$}$ from $\PP{m}{\$}$ and $\LL{m}{m}{\empty}$ from $\PP{m}{\empty}$.
\end{lemma}
%%%\vspace{-2mm}  

\subsection{Proof of Lemma~\ref{Lmm}}\label{Lmm-proof}
\begin{proof}
That is an easy consequence of Lemmas~\ref{L1}, \ref{layer}, \ref{nodolar} and the definition of $S_{layer}^{m}$.
\end{proof}

%%%\vspace{-2mm}
\begin{observation}
By Exercise~\ref{iso} Lemma~\ref{Lmm} proves Lemma~\ref{grid}. 
\end{observation}
%%%\vspace{-2mm}

%%%%%%%%%%%%%%%%%%%%%%%%%%%  Principle II $$$$$$%%%%%%%%%%%%%%%%%%%%%%%%%%%%%%%%

%%%%%%%%%%%%%%%%%%%%%%%%%%%%%%%%%%%%%%%%%%%%%%%%%%%%%%%%%%%%%%%%%%%%%%%%%%%%%%%%
%%%%%%%%%%%%%%%%%%%%%%%%%%%%%%%%%%%%%%%%%%%%%%%%%%%%%%%%%%%%%%%%%%%%%%%%%%%%%%%%

%%%\vspace{-2mm}  
\section{And now we finally see the shades again}\label{ulga}
%%%\vspace{-2mm}  

Now we are ready to finish the proof of Lemma \ref{poprawnoscredukcji}. 
First assume the original instance of Our Grid Tiling Problem has no \textit{proper shading}.  
The following is straightforward from König's Lemma by noticing that if there were arbitrary grids with proper shading, then there would be an infinite one:
%%%\vspace{-2mm}  
\begin{lemma}
\label{konig}
If an instance $I$ of OGTP has no \textit{proper shading} then there exist natural $m$ such that for any $k \geq m$ a square grid of size $k$ has no shading that satisfies conditions (a1), (a2), (b1) and (b3) of \textit{proper shading}.
\end{lemma}
%%%\vspace{-2mm}  

Let $m$ be the value from Lemma~\ref{konig}. By Lemma~\ref{grid} {\em Crocodile} can force {\em Fugitive} to build a structure isomorphic to either $\GG{m+1}{\empty}$ or $\GG{k}{\$}$ for some $k \leq m$. Now suppose the play ended, in some final position ${\mathbb H}$ isomorphic to one of these structures. We take off our glasses, and not only we still see this ${\mathbb H}$, but now we also see the shades, with each edge (apart from edges labeled with $\alpha, \omega, x$, $y$ and \$) having one of the shades from $\mathcal S$. Now concentrate on the red edges labeled with $\CC{\bullet}{\empty}{W}{\empty}{\empty}$ of  ${\mathbb H}$. They form a grid, with each vertical edge labeled with $V$, each horizontal edge labeled with $H$, and with each edge labeled with a shade from $\mathcal S$.
Now we consider two cases:
\begin{itemize}
\item If $\GG{m+1}{\empty}$ was built then clearly condition (b3) of Definition \ref{shading} is unsatisfied. But this implies that a path labeled with a word from one of the languages $Q_{bad}$ occurs in $\mathbb H$ between $a$ and $b$, which is in breach with Principle III because of language $Q_{bad}^{1}$.
\item If $\GG{k}{\$}$ for $k \leq m$ was built then clearly condition (b2) or (b3) of Definition \ref{shading} is unsatisfied. This is because we assumed that there is no \textit{proper shading}. But this implies that a path labeled with a word from one of the languages $Q_{bad}$ occurs in $\mathbb H$ between $a$ and $b$, which is in breach with Principle III because of language $Q_{bad}^{1}$. 
\end{itemize}
This ends the proof of Lemma \ref{poprawnoscredukcji} (ii).

For the proof of Lemma \ref{poprawnoscredukcji} (i) assume the original instance $\langle {\mathcal S}, {\mathcal F}\rangle$ of Our Grid Tiling Problem has a \textit{proper shading} -- a labeled grid of side length $m$. Call this grid ${\mathbb G}$.

Recall that $\GG{m}{\$}$ satisfies all regular constraints from ${\mathcal Q}_{good}^\leftrightarrow$ and from ${\mathcal Q}_{ugly}^{\leftrightarrow}$ (Exercise~\ref{noreqG}). Now copy the shades of the edges of ${\mathbb G}$ to the respective edges of $\GG{m}{\$}$. Call this new structure ($\GG{m}{\$}$ with shades added) $\mathbb M$. It is easy to see that 
$\mathbb M$ constitutes a counterexample, as in Lemma \ref{lm-det-struct}.\\ 

%%%%%%%%%%%%%%%%%%%%%%%%%%%%%%%%%%%%%%%%%%%%%%%%%%%%%%%%%%%%%%%%%%%%%%%%%%%%%%%%

%%%%%%%%%%%%%%%%%%%%%%%%%%%%%%%%%%%%%%%%%%%%%%%%%%%%%%% REFERENCES %%%%%%%%%%%%%%%%%%%%%%%%%%%%%%%%%%%%%%%%

\bibliographystyle{acm}
\bibliography{references}

\end{document}